\documentclass[acmsmall,nonacm]{acmart}

\acmConference 

\AtBeginDocument{%
  \providecommand\BibTeX{{%
    Bib\TeX}}}

\usepackage{xcolor}

\usepackage{fontawesome,bm}
\usepackage{amsthm}
\usepackage{mathtools}
\usepackage[linesnumbered,algonl]{algorithm2e}
\usepackage{dsfont}
\usepackage{enumitem}
\usepackage{cleveref}
\usepackage{optidef}
\usepackage{subcaption}

\newtheorem{theorem}{Theorem}

\newtheorem{lemma}{Lemma}
\newtheorem{definition}{Definition}
\theoremstyle{definition}
\newtheorem{remark}{Remark}

\usepackage{xspace}

\newcommand{\PHASES}{\Phi}

\DeclarePairedDelimiter{\ceil}{\lceil}{\rceil}

\newcommand{\avgpacks}{\lambda}

\newcommand{\be}{\begin{eqnarray}}

\newcommand{\ee}{\end{eqnarray}}
\newcommand{\ben}{\begin{eqnarray*}}
\newcommand{\een}{\end{eqnarray*}}
\newcommand{\bfl}{\begin{flalign*}}
\newcommand{\efl}{\end{flalign*}}
\newcommand{\FI}{0}
\newcommand{\SI}{1}
\newcommand{\TI}{2}

\newcommand{\USERS}{\MATH{\mathcal K}}

\newcommand{\BSTATE}[1]{
\ifx&#1&%
   \MATH{\Psi}
\else
   \MATH{\Psi(#1)}
\fi
}

\newcommand{\TFSTATE}[1]{
\ifx&#1&%
   \MATH{\mathbf z}
\else
   \MATH{\mathbf z(#1)}
\fi
}

\newcommand{\LP}[0]{\texttt{LP}\xspace}

\newcommand{\MATH}[1]{\ensuremath{#1}\xspace}
\newcommand{\br}[1]{\left[#1\right]}

\newcommand{\ES}{\mathbb E}

\newcommand{\E}[1]{\mathbb E \br{#1}}

\newcommand{\ALG}{\ensuremath{\mathrm{ALG}}\xspace}

\usepackage{dsfont,bm}

\newcommand{\sumtlast}{\sum_{\trl=(t-\DMAX)^+}^{t \land \TC }}

\newcommand{\IF}{C}
\newcommand{\ACTIVELINK}[2]{\IF_{#1}(#2)=1}

\newcommand{\ACTIVES}[1]{\{l\in \USERS: \ACTIVELINK{l}{t}\}}

\newcommand{\dref}[1]{(\ref{#1})}

\newcommand{\sumajt}{\sum_{n=1}^{\ajt}}
\newcommand{\sumajtrl}{\sum_{n=1}^{\ajtrl}}
\usepackage{amsfonts}

\newcommand{\ALGOFF}{\ensuremath{\mathrm{FBPF}}\xspace}
\newcommand{\ALGOFFSPELLED}{Flow-Based Probabilistic Forwarding\xspace}
\newcommand{\ALGON}{\ensuremath{\mathrm{DLPF}}\xspace}
\newcommand{\ALGONSPELLED}{Dynamic Learning with Probabilistic Forwarding\xspace}

\newcommand{\tprm}{t^\prime}
\newcommand{\arv}[2]{a_{#1}^{#2}}
\newcommand{\aj}[1]{\arv{j}{#1}}
\newcommand{\ajt}{\aj{t}}
\newcommand{\ajtrl}{\aj{\trl}}
\newcommand{\fl}[3]{f_{#1#2}^{#3}}
\newcommand{\flo}[3]{f_{#1#2}^{#3\star}}

\newcommand{\trl}{\tau}
\newcommand{\PKT}{\ensuremath{\mathrm J}\xspace}
\newcommand{\PD}{\ensuremath{[\PKT]}\xspace}

\newcommand{\AMAX}{\ensuremath{a_{max}}\xspace}

\newcommand{\DMAX}{\ensuremath{d_{max}}\xspace}
\newcommand{\WMAX}{\ensuremath{w_{max}}\xspace}

\newcommand{\LNKS}{\ensuremath{\mathcal L}\xspace}
\newcommand{\V}{\ensuremath{\mathcal V}\xspace}
\newcommand{\G}{\ensuremath{\mathcal G}\xspace}
\newcommand{\TC}{\ensuremath{\mathrm T}\xspace}
\newcommand{\TCP}{\ensuremath{\mathrm T^\prime}\xspace}
\newcommand{\RI}{\ensuremath{\mathrm {RI}(\TC)}\xspace}
\newcommand{\ERI}{\ensuremath{\mathbb E \big[W_{\RI}^\star\big]}}

\newcommand{\EI}{\ensuremath{\mathrm {EI}(\TC)}\xspace}
\newcommand{\EIB}{\ensuremath{\mathrm {\overline {EI}}(\TC)}\xspace}
\newcommand{\EST}{\ensuremath{\mathrm E_S}\xspace}
\newcommand{\FS}{\ensuremath{\mathrm F_S}\xspace}
\newcommand{\FNS}{\ensuremath{\mathrm F_{NS}}\xspace}
\newcommand{\TB}{\ensuremath{\overline{\mathrm T}}\xspace}
\newcommand{\EG}{\ensuremath{\overline {\mathcal G}}\xspace}

\newcommand{\ELNKS}{\ensuremath{\overline {\mathcal L}}\xspace}

\newcommand{\WMX}{\ensuremath{w_{\max}}\xspace}
\newcommand{\WMN}{\ensuremath{w_{\min}}\xspace}
\newcommand{\sumT}{\sum_{t=1}^{\TC}} 
\newcommand{\sumk}{\sum_{k\in \mathcal K_j} } 
\newcommand{\sumpkt}{\sum_{j=1}^{\PKT}} 
\newcommand{\sumtotarrv}[1]{ \sumT \sumpkt} 

\newcommand{\sumsched}{\sum_{k\in \mathcal K_j} }
\newcommand{\sumschedl}{\sum_{k: k_{t-\trl}=\ell} }
\newcommand{\vRI}[1]{y_{jk}^{#1}}
\newcommand{\vRIO}{\vRI{nt\star}}
\newcommand{\vRIOtrl}{\vRI{n\trl\star}}

\newcommand{\vEI}[1]{x_{jk}^{#1}}

\newcommand{\sumKK}{\sum_{k \in \mathcal K_{j}:k_{\trl}=l}}
\newcommand{\sumTT}{\sum_{\trl=0}^{\DMAX}}

\newcommand{\ahat}{\hat \lambda_{j}}
\newcommand{\abar}{\lambda_{j}}
\newcommand{\amin}{\lambda_{\min}}
\newcommand{\abart}{\lambda_{j}^t}
\newcommand{\abartau}{\lambda_{j}^{\trl}}
\newcommand{\C}{\mathrm {C}}
\newcommand{\LL}{L}
\newcommand{\CMN}{\ensuremath{\C_{\min}}\xspace}
\newcommand{\CMX}{\ensuremath{\C_{\max}}\xspace}
\newcommand{\AO}[1]{\ensuremath{\mathtt{Out}(#1)}}
\newcommand{\AI}[1]{\ensuremath{\mathtt{Inc}(#1)}}

\def\BibTeX{{\rm B\kern-.05em{\sc i\kern-.025em b}\kern-.08em
    T\kern-.1667em\lower.7ex\hbox{E}\kern-.125emX}}
 
\begin{document}

\title[Near-Optimal Packet Scheduling  in Multihop Networks with End-to-End Deadline Constraints]{Near-Optimal Packet Scheduling in Multihop Networks with End-to-End Deadline Constraints}

\author{Christos Tsanikidis}
\email{c.tsanikidis@columbia.edu}
\affiliation{%
  \institution{Columbia University}
  \city{New York}
  \state{NY}
  \country{USA}}

\author{Javad Ghaderi}
\affiliation{%
  \institution{Columbia University}
  \city{New York}
  \state{NY}
  \country{USA}}
\email{jghaderi@columbia.edu}

\begin{abstract}
Scheduling packets with end-to-end deadline constraints  in multihop networks is an important problem that has been notoriously difficult to tackle. Recently, there has been progress on this problem in the worst-case traffic setting, with the objective of maximizing the number of packets delivered within their deadlines. 
Specifically, the proposed algorithms were shown to achieve  $\Omega(1/\log(\LL))$ fraction of the optimal objective value if the minimum link capacity in the network is $\CMN=\Omega(\log (\LL))$, where $\LL$ is the maximum length of a packet's route in the network (which is bounded by the packet's maximum deadline). However, such guarantees can be quite pessimistic due to the strict worst-case traffic assumption and may not accurately reflect real-world settings. 
In this work, we aim to address this limitation by exploring whether it is possible to design algorithms that achieve a constant fraction of the optimal value while relaxing the worst-case traffic assumption. 
 We provide a positive answer by demonstrating that in stochastic traffic settings, such as i.i.d. packet arrivals, near-optimal, $(1-\epsilon)$-approximation algorithms can be designed if $\CMN = \Omega\big(\frac{\log (\LL/\epsilon) } {\epsilon^2}\big)$. To the best of our knowledge, this is the first result that shows this problem can be solved near-optimally under nontrivial assumptions on traffic and link capacity. We further present extended simulations using real network traces with non-stationary traffic, which demonstrate that our algorithms outperform worst-case-based algorithms in practical settings.
\end{abstract}

\maketitle

\section{Introduction}

In recent years, the scheduling of real-time traffic in communication networks has become increasingly important, due to growing number of emerging real-time applications such as video streaming and video conferencing, vehicular networks, cyber-physical networks, and Internet-of-Things \cite{popovski2022perspective}.
The connectivity within these networks is often complex, with packets generated at specific sources requiring the traversal of multiple links in order to reach their destinations successfully.
Additionally, in many real-time applications, the timely delivery of packets is a primary concern, as packets that fail to meet specific deadlines, on the total time from generation of
a packet at its source until delivery to its destination, are typically discarded by the application.  Meeting the deadline constraints requires a departure from traditional scheduling algorithms (e.g. based on MaxWeight or Backpressure~\cite{tassiulas1992stability}) that have been designed for maximizing throughput and cannot provide deadline guarantees on packet delivery.

Despite the importance of the problem due to its broad applicability, there is very limited work on techniques with attractive theoretical guarantees in multihop networks. 
This is due to the complexity of the problem, including the need to make online decisions, the exponential growth in the number of scheduling decisions that involve the path a packet
takes as well as the specific time slots the packet occupies for
transmission on the links, and the stringent deadline constraints.
In particular, prior literature on scheduling packets with deadlines in multihop networks has focused on either worst-case traffic \cite{tsanikidis2022online,gu2021asymptotically,deng2019online,mao2014optimal} with pessimistic approximation ratios, or stochastic traffic with guarantees in the case that the link capacities are relaxed to be only satisfied on average~\cite{singh2018throughput},
or in an asymptotic regime where the link capacities and packet arrival rates scale to infinity~\cite{singh2018throughput}. As a result, an important open question remains: \textit{is it possible to design algorithms that have attractive performance guarantees (e.g. constant approximation ratio that does not depend on parameters of the network), for finite link capacities and given stochastic traffic}? In this work, we answer this question positively by providing algorithms that are near-optimal while only requiring minimum link capacities that are logarithmic in the maximum length of  a packet's route $L$ (with $L\leq \DMAX$, where $\DMAX$ is the maximum deadline of any packet). Further, our algorithms have polynomial computational complexity, are amenable to distributed implementations, and can be easily adapted to different traffic distribution assumptions. Additionally, our work 
features distinct techniques compared to prior work in the field 
\cite{gu2021asymptotically,deng2019online,singh2018throughput, tsanikidis2022online}.

The near-optimality of our algorithms is with respect to the commonly studied average approximation ratio, which is the fraction of the optimal objective value our algorithm obtains on average (this modifies the worst-case approximation ratio in \cite{tsanikidis2022online,deng2019online,gu2021asymptotically} to a stochastic setting)\footnote{In the literature, similar results are presented in terms of the \textit{competitiveness}, which is the reciprocal of the approximation ratio, that is to say, $a$-approximation ratio corresponds to a $1/a$-competitiveness.}. In the weighted-packet case, unlike prior techniques, e.g.~\cite{tsanikidis2022online,gu2021asymptotically}, our results do not depend on the weights of the packets. Additionally, our results hold for finite time horizons, as opposed to prior works in the stochastic setting that require infinite time horizons, e.g.~\cite{singh2018throughput,singh2013pathwise}.

In our view, this work fills an important gap in the literature regarding the feasibility of solving the real-time scheduling problem near-optimally in practical multihop networks which are characterized by finite link capacities, finite time horizons, and stochastic traffic. 

\subsection{Related Work}
The related work can be divided into three categories: scheduling traffic with deadlines in single-hop networks, scheduling traffic with deadlines in multi-hop networks, and works on online stochastic programming. 

\textbf{Single-hop networks}. In single-hop networks, traffic between any source-destination node needs to traverse only one link. If the network has a single link, there is extensive literature that guarantees a constant approximation ratio (for example $\frac{e-1}{e}$), e.g.,~\cite{chin2006online,jez2013universal}. 
Further, when all packets can be successfully delivered within their deadlines (called the \textit{underloaded} regime), simple algorithms such as Earliest-Deadline-First become optimal \cite{baruah1992competitiveness,liu1973scheduling}.
In the absence of interference among links, these techniques can be applied to wired networks, as each link can be handled independently due to the one-hop traffic. However, in networks that suffer from interference between links, the required solutions often become more involved, as the decisions across links become coupled. The problem has been studied in a variety of works, e.g., \cite{hou2009theory,hou2013scheduling,kang2013performance,tsanikidis2020power,tsanikidis2021randomized,singh2013pathwise} under different traffic considerations, benchmarks, and types of interferences. Many of the prior works consider frame-based traffic, in which, packets usually are assumed to arrive at the beginning of the frame and expire at the end of the frame (e.g. \cite{hou2009theory,hou2013scheduling}). Recently, advancements on the problem in the more general traffic case have been made \cite{kang2014performance,tsanikidis2020power}, specifically, it has been shown that obtaining constant fractions of the optimal objective value is still possible, with approximation ratios ranging from $0.5$ to $\frac{e-1}{e}$, and complexity that depends on the interference graph \cite{tsanikidis2020power, tsanikidis2022online}\footnote{Although the performance metric in many of these  works is in terms of the \textit{real-time capacity region}, this metric can be related to the approximation ratio we study here \cite{tsanikidis2021randomized}.}.
 
\textbf{Multi-hop networks.}
The problem becomes significantly more challenging when multihop traffic is considered. The past work can be divided into four groups: (i) heuristics without theoretical guarantees (e.g., \cite{li2012scheduling,liu2019spatial}), (ii) algorithms for restrictive topologies such as single-destination tree \cite{bhattacharya1997optimal,mao2014optimal}, 
(iii) approaches that provide approximation ratios that diminish as parameters of the network scale
~\cite{li2012scheduling,li2009minimizing,liu2019spatial,wang2011end, mao2014optimal,deng2019online,gu2021asymptotically,tsanikidis2022online}, and (iv) techniques with guarantees for relaxations of the problem \cite{andrews1999packet,andrews2000general,singh2018throughput} or when certain parameters of the network are scaled to infinity \cite{singh2018throughput}. Below, we highlight works with theoretical guarantees in the latter two groups, which are more relevant to our work.

The recent works in~\cite{deng2019online,gu2021asymptotically,tsanikidis2022online} provide the best existing guarantees for the problem in the worst-case traffic setting. The work~\cite{deng2019online} studied the problem without any packet weights, and provided the first algorithm that can achieve $\Omega(1/\log L)$ approximation,  when the minimum link capacity grows to infinity ($C_{\min} \rightarrow \infty$). This result matches the lower-bound provided in \cite{mao2014optimal} and hence is asymptotically optimal for the worst-case traffic.  
The work~ \cite{gu2021asymptotically} introduced several variants of an algorithm called GLS, which in the best case, with $\CMX=\CMN=\Omega(\log L)$, yield a $\Omega(1/\log L)$-approximation as well. 
 This line of research was improved further in \cite{tsanikidis2022online} to more general and improved techniques for weighted packets. In particular, for $\CMN=\Omega( \log (\rho L))$ an algorithm with $\Omega(1/\log (\rho L))$ approximation ratio was provided, where $\rho=\frac{\WMX}{\WMN}$ is the ratio of maximum to minimum packet weight, $\WMX$ and $\WMN$.


In addition to the work in the worst-case traffic setting, there is work~\cite{singh2018throughput} in the case of stochastic traffic, which considers a relaxed version of the problem under \textit{average} link capacity constraints. However, this relaxation considerably simplifies the problem. In fact, some of the primary challenges of the studied problem in our work, revolve around handling the strict capacity constraints as opposed to average constraints. The work~\cite{singh2018throughput} shows that the performance loss due to such relaxation becomes asymptotically small in the regime that the link capacities and arrival rates are all scaled to infinity.  However, the performance in real networks with finite capacities and finite arrival rates is not clear.


There is also work that
relaxes the notion of strict deadlines considered in our paper, e.g., \cite{andrews2000general,andrews1999packet}. 
The work in~\cite{andrews2000general} focuses on underloaded networks, with fixed routes, where the total packet arrival rate to each link is less than the link's capacity (assumed to be 1), and packets \textit{do not} have strict deadlines. In this case, all packets can be delivered in a bounded time. The paper proposes algorithms with guaranteed bounded delay as a function of the arrival rates and the network size. 
The work in~\cite{andrews1999packet} considers deadlines but they are allowed to be violated by some factor. The system is again assumed to be underloaded. It provides algorithms where the violation factor is bounded by a function of arrival rates and the number of links in the network.   
In contrast, in our work, packets exceeding their deadlines have no utility and therefore are discarded. Further, we do not assume an underloaded system, e.g., the system might be overloaded, in which case, even the optimal algorithm has to drop some packets. Finally, we consider different weights or rewards for packets as opposed to unweighted packets in~\cite{andrews1999packet,andrews2000general}.  
Note that the underloaded system assumption considerably simplifies the problem. For example, in the case of single-hop networks with unit capacity, as mentioned earlier, Earliest-Deadline-First is optimal when the system is underloaded. However, when packets have weights and the system is overloaded, the approximation ratio achieved by any online algorithm is strictly less than one~\cite{jez2013universal}.


Although the above works have made significant advances on the problem, they raise the following lingering question: Can we do better than an approximation ratio that becomes increasingly worse as network parameters, such as $L$, become larger, in the case of finite link capacities and for the given traffic? In this work, we address this question by providing algorithms that, in the presence of stochastic traffic, can provide \textit{near-optimal} performance.

\textbf{Online Stochastic Programming.} Beyond the above related literature, there has been a line of research \cite{devanur2019near,agrawal2014fast,kesselheim2014primal,banerjee2020uniform} for online and sequential assignment of limited \textit{resources} to \textit{requests}.
Requests are drawn from an i.i.d. distribution of \textit{request types}, which characterize how these requests can be assigned to resources. Then, each request of a given request type can be assigned to the same fixed set of resources regardless of the time of arrival. While there is similarity between this model and our problem, for example, by considering ``requests'' as the arriving packets, and ``resources'' as each pair of (link, time slot) combinations, unfortunately, however, these techniques cannot be applied to our problem. This is because these works assume that requests of a given request type can be assigned to the same resources, whereas, in our problem, packets of a given packet type cannot be assigned to the same (link, time slot) pair. For example, a packet arriving at time $t'$ after $t$ cannot be assigned to the (link, time slot) pair $(l,t)$. Further, these techniques often require randomizing across all decision variables, which in our case, is prohibitive, as there are exponentially many options for scheduling a packet in the network. 

Finally, we point out that there is a rich literature on the broader problem of time-sensitive scheduling and/or routing, e.g., \cite{leighton1994packet,wang2014energy,sun2021age,srinivasan1997constant,fountoulakis2023scheduling,fan2019approach}. In ~\cite{leighton1994packet}, the time to schedule a set of packets (makespan) in a network was studied. Deadline-constrained scheduling has also been studied jointly with other objectives such as energy consumption~\cite{wang2014energy,fan2019approach} or the age of information~\cite{sun2021age,fountoulakis2023scheduling}.

\subsection{Contributions} The main contributions of this paper can be summarized as follows.

\textbf{Near-Optimal Static Algorithm.} We design the first algorithm, to the best of our knowledge, that guarantees \textit{near-optimal} performance for scheduling packets arriving as a stochastic process with arbitrary hard deadlines in multihop networks with \textit{finite and strict} link capacities. Assuming the knowledge of the packet arrival rates, our algorithm relies on solving a single linear program with polynomial number of variables and constraints. Subsequently, packets are treated independently from each other and are scheduled over different links according to forwarding probabilities calculated through the solution of the linear program. Our algorithm provides $(1-3\epsilon)$-approximation when the minimum link capacity $\CMN$ satisfies $\CMN \geq 2 \left(\frac{ 1+\epsilon}{\epsilon}\right)^2 \log (L/\epsilon)$, when packet arrivals are generated by a set of Bernoulli (or Binomial) processes, with any deadline and weight. Our result does not require relaxing the capacity constraints or taking a limit on arrival rate or capacity. Furthermore, our guarantees hold for any finite time horizon of length $\TC\geq \frac{2 } {\epsilon}\DMAX^2$.


\textbf{Near-Optimal Dynamic Algorithm.} When the knowledge of the packet arrival rates is not available, we provide a dynamic algorithm by dividing time into $\ceil{\log(1/\epsilon)}$ phases of geometrically increasing lengths. At the beginning of each phase, the packet arrivals in prior phases are used to estimate the packet arrival rate of different packet types (flows), which are subsequently used in the linear program by the static algorithm for the incoming phase. Our dynamic algorithm preserves the $(1-O(\epsilon))$-approximation, when the time horizon is larger than a certain threshold.

\textbf{Extensions to Non-Stationary and Dependent Traffic.}
Our 
techniques are versatile and can be adapted to other distributional assumptions. For instance, in the case that packet arrivals  exhibit dependence across time slots or within a time slot, our method achieves near-optimal performance when $\CMN \geq 2 D \left(\frac{ 1+\epsilon}{\epsilon}\right)^2 \log (L/\epsilon)$, where $D$ represents the degree of dependence among packets (in the i.i.d. scenario, $D=1$). This allows extending our results to more general distributions on the number of arrivals per time slot. Additionally, our techniques can be extended to periodic traffic distributions and non-stationary distributions. 
The combination of these extensions results in a wide range of stochastic arrival processes that can be effectively addressed through our method.

\textbf{Empirical Evaluation Using Real Datasets.} 
In order to evaluate the performance of our algorithms in practical settings, in addition to extensive synthetic simulations, we provide simulations using real traffic traces, and over real networks.
The results indicate that, despite the presence of highly non-stationary traffic in the network traces, our algorithms demonstrate significant performance improvement over the worst-case-based algorithms.

\subsection{Notations}
We use $[n]$ to denote the set $\{1,2,\cdots,n\}$. Further we denote $[n]_0:=[n]\cup\{0\}$. We use $\mathbb N := \{1,2,3,\cdots\}$, and $\mathbb R^+:=\{ x\in \mathbb R|x\geq 0\}$. We define $a \land b:=\min \{a,b\}$, and $(a)^+:=\max(a,0)$. 


\section{Model and Definitions}\label{sec:model}

Consider a communication network consisting of a set of nodes $\V$ and a set of communication links $\LNKS$ between the nodes, defining a directed graph $\G=(\V,\LNKS)$. We assume that time is slotted, i.e., $t = 1, 2,3, \cdots$. Each link  $\ell=(u,v)\in \LNKS$ has a capacity $\C_\ell$ which is the maximum number of packets per time slot that can be transmitted over the link from node $u \in \V$ to node $v \in \V$.
To simplify future discussions, we further define $\EG:=(\V,\ELNKS)$ to denote 
$\G$ with the addition of self-loops to the set of edges, i.e., $\ELNKS=\LNKS \cup \{ (u,u): u\in \V\}$. 
Scheduling a packet over a ``self-loop link'' then has the interpretation that the packet remains at the same node, and we define $\C_{\ell}=\infty$ for any self-loop link $\ell$. 
We further use $\AO{v}$ to denote the set of outgoing links of a node $v$ in $\EG$, i.e., $\AO{v}=\{(v,u): (v,u)\in \ELNKS\}$, and similarly use $\AI v$ to denote its incoming links, i.e., $\AI v=\{(u,v): (u,v)  \in \ELNKS\}$. Note that by definition, self-loop link $(v,v)$ belongs to both $\AO v$ and $\AI v$.

The network is shared by a set of packet types (flows) $\PD:=\{1,2,\cdots,\PKT\}$. A packet type (flow) $j \in \PD$ is characterized by its source $s_j \in \V$, its destination $z_j \in \V$, an end-to-end deadline $d_j \in \mathbb N\cup\{0\}$, and a weight $w_j \in \mathbb R^+$. Packets of different types arrive during a time horizon of length $\TC$.
A packet of type $j$ arriving at the beginning of time slot $t$ has to reach its destination before the end of time slot $d_j+t$ to yield a reward $w_j$, otherwise it is discarded. Packets of type $j \in \PD$ arrive according to a stochastic process $\{\ajt\}_{ t\in [\TC]}$ with possibly time-dependent arrival rate $\abart$, where $\ajt \geq 0$ is the number of packet arrivals of type $j$ at time $t$, and $\abart=\ES[{a_j^t}]$.
We denote the maximum deadline, the maximum number of arrivals, and the minimum arrival rate of any packet type by $\DMAX:=\max_{j \in \PD} d_j$, $\AMAX:=\max_{t\in [\TC],j \in \PD} a_j^t$ and $\amin:=\min_{t\in [\TC],j \in \PD} \abart$, respectively.

A packet of type $j$ arriving at time $t$, which is scheduled over the network, should be transmitted over a sequence of links, at specified time slots, that deliver it from its source $s_j$ to its destination $z_j$ before the end of time slot $t+d_j$. This sequence of links and time slots is characterized through the notion of \textit{relative route-schedule}, defined in \Cref{define-route-schedule} below.
\begin{definition}[Relative Route-Schedule]\label{define-route-schedule}
A (relative) route-schedule $k$ for a packet of type $j$ is a directed walk with sequence of edges
$
[k_{\FI}\quad k_{\SI}\quad \cdots \quad k_{d_j}]
$
in graph $\EG$, where $k_{\trl} \in \ELNKS$ is the link over which the packet is scheduled at the $\trl$-th time slot following its arrival, with $k_{\FI} \in \AO {s_j}$ and $k_{d_j} \in \AI {z_j}$. We use $\mathcal K_j \subseteq \ELNKS^{d_j+1}$ to denote the set of all valid route-schedules for packet-type $j \in \PD$.
\end{definition}

For notational convenience, we define  $k_{\trl}=\varnothing$ for $\trl>d_j$, where $k\in \mathcal K_j$.
Our objective is to maximize the weighted sum of the packets that are successfully delivered from their sources to their destinations within their deadlines. Given a random instance of the packet arrival sequence, this optimization can be formulated as an integer program over the time horizon $\TC$, defined in \dref{RIT-eq-first}-\dref{RIT-eq-last} below. We refer to this optimization as $\RI$ which stands for \textit{Random Integer problem over the time horizon of length $\TC$}. In $\RI$, the optimization is over the scheduling decisions $\mathbf y= \{\vRI {nt}\}$, where each $\vRI {nt}$ is a binary variable indicating whether the $n$-th arriving packet of type $j$ at time $t$ is scheduled using a relative route-schedule $k \in \mathcal K_j$ or not.
\begin{subequations}\label{RIT}
\begin{align}
\max_{\mathbf{y}} \quad & 
 \sumtotarrv{t} w_j \sum_{n=1}^{\ajt} \sumsched \vRI {nt} \quad\quad\quad (:=\RI) \label{RIT-eq-first}\\ 
\textrm{s.t.} \quad & 
\sumk \vRI {nt}  \leq 1,\quad \forall t\in [\TC],\forall j\in \PD, \forall n \in [\ajt],\label{RIT:one-choice}\\
&  \sumtlast
\sumpkt \sum_{n=1}^{\ajtrl}
\sumschedl \vRI {n\trl} \leq \C_{\ell},\quad \forall \ell\in \LNKS, \forall t\in [\TB], \label{RIT:capacity} \\
& \vRI {nt} \in \{0,1\}, \quad \forall j\in \PD, \forall t\in [\TC],\forall k \in \mathcal K_j.  \label{RIT-eq-last}
\end{align}
\end{subequations}

The objective function \dref{RIT-eq-first} is the sum of the weights of arriving packets which are successfully scheduled in the network.
Constraints \dref{RIT:one-choice} and \dref{RIT-eq-last} require that each arriving packet is assigned to ``at most'' one route-schedule (or not scheduled at all). Constraints \dref{RIT:capacity} require each link's capacity to be enforced at every time slot $t\in [\TB]$ that packets can exist in the system, where $\TB:=\TC+\DMAX$.
Specifically, in the left-hand-side of \dref{RIT:capacity}, we count the number of packets that are scheduled to be transmitted on link $\ell$ at time $t$. By the definition of deadline, these packets could have only arrived in the last $\DMAX+1$ time slots, i.e., at times $\trl \in \{ t, t-1,\cdots, t-\DMAX\}$, and are counted in the left-hand-side of \dref{RIT:capacity} if they are scheduled using some route-schedule $k$ which, $t-\trl$ slots after the packet arrival at time $\trl$, transmits it on $\ell$ at time $t$ (i.e. the route schedules with $k_{t-\trl}=\ell$). 
Refer to \Cref{flow-route-example-1} for an illustration of the notations through an example.

\begin{figure}[t]
\centering

\begin{subfigure}{.35\textwidth}
  \centering
  \includegraphics[width=0.90\linewidth]{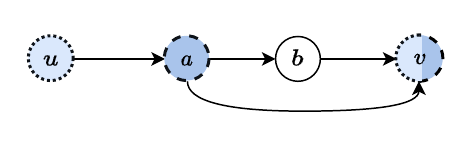}
  \caption{Example network graph $\G$}
  \label{fig:example-graph}
\end{subfigure}%
\begin{subfigure}{.33\textwidth}
  \centering
  \includegraphics[width=0.90\linewidth]{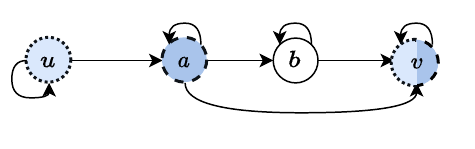}
  \caption{Graph $\EG$ with self-loops}
  \label{fig:example-loop}
\end{subfigure}

\caption{A network $\G$ with nodes $\{u,a,b,v\}$ and $4$ links with unit capacity, 
and the corresponding graph $\EG$ with self-loops. Suppose there are 2 packet types $\PD=\{1,2\}$  with sources $s_1=u,s_2=a$, destinations $z_1=z_2=v$, and deadlines $d_1=2$ and $d_2=1$.
Two possible route-schedules for packet-types $1$ and $2$ respectively are: $k=[(u,u)\quad (u,a)\quad({a},{v})]$ and $k^\prime=[({a},v)\quad (v,v)]$. Suppose there are arrivals of type $1$ and type $2$ at time slots $t_0$ and $t_0+2$ respectively (i.e. $\arv{1}{t_0}=\arv{2}{t_0+2}=1$), then, scheduling both packets on $k,k^\prime$ (i.e. $y_{1k}^{1t}=y_{2 k^\prime}^{1(t+2)}=1$) is not possible due to insufficient
capacity on link $(a, v)$ at time $t+2$ (violating constraint \dref{RIT:capacity}).}
\label{flow-route-example-1}
\end{figure}

Our goal is to provide algorithms that guarantee good performance in the average sense, as formalized through the definition below.

\begin{definition}
Suppose the optimal objective value of $\RI$ is $W_{\RI}^\star$. An algorithm $\ALG$ provides $\gamma$-approximation to $\RI$, if the objective value attained using $\ALG$, $W_{\ALG}$, satisfies:
\[\mathbb E \big[W_{\ALG}\big]\geq \gamma \ERI,\]
where the expectation in \ERI is with respect to the randomness in the random instance $\RI$ (i.e., with respect to the arrival sequence $\{\ajt\}$), and $\E{W_{\ALG}}$ is with respect to the randomness in the random instance, and, if applicable, the random decisions of $\ALG$.
\end{definition}

\section{Algorithms and Main Results} \label{sec:algo-results}
In this section, we introduce the two main algorithms under the assumption of fixed packet arrival rates, i.e., $\avgpacks_{j}^t \equiv \avgpacks_j$. Extensions to this assumption are discussed in \Cref{generalization-non-stationary}. We first introduce \Cref{randomized-scheduling-offline} (\textit{\ALGOFFSPELLED}) for the case of known packet arrival rates, and subsequently \Cref{randomized-scheduling-online} (\textit{\ALGONSPELLED}) for the case where the packet arrival rates are unknown. For each algorithm, we state its corresponding approximation ratio.

\subsection{\texorpdfstring{\Cref{randomized-scheduling-offline}: \ALGOFFSPELLED}{}}
We first introduce \Cref{randomized-scheduling-offline}, a scheduling algorithm that probabilistically forwards packets at each time slot based on their \textit{age} (time since their arrival), \textit{type}, and their \textit{current node} in the system. Specifically, \Cref{randomized-scheduling-offline} forwards a type-$j$ packet, which is currently at node $v$, over an outgoing link  $\ell \in \AO v$ with a probability determined by forwarding variables $f_{j\ell}^\trl$, which depend on the packet type $j$ and the number of time slots $\trl \in \{0,1,\cdots, d_j\}$ that the packet has been in the system so far (i.e., its age). 
To determine the forwarding probabilities, \Cref{randomized-scheduling-offline} solves a Linear Program (\LP), as a preprocessing step.
In the following, we first explain the \LP and subsequently discuss the forwarding process based on its solution.  

The forwarding variables $\mathbf f=\{f_{j\ell}^\tau\}$ are computed as the solution to the following \LP, which we refer to as $\FS$:
\begin{subequations}\label{FS}
\begin{align}
\max_{\mathbf{f}} \quad & 
 \sum_{j=1}^{\PKT} w_{j} \avgpacks_{j} \sum_{\ell \in \AO {s_j}} \fl{j} \ell {\FI}
\quad(:=\FS) \label{FS-obj} \\ 
\textrm{s.t.} \quad & 
\sum_{\ell \in \AO{s_j}} \fl j \ell \FI \leq 1,\quad \forall j\in \PD, \label{FS:one-choice}\\
& 
\sum_{\ell \in \AI{v}} \fl{j}{\ell}{\trl-1}=\sum_{\ell \in \AO{v}} \fl j \ell \trl,\quad \forall v\in \V, \trl \in [d_j], \forall j\in \PD, \label{FS:time-conserv}\\
&f_{j\ell}^{\FI}=0,\ \ \forall \ell\not \in \AO{s_{j}}, 
\quad f_{j\ell}^{d_j}=0,\ \ \forall \ell\not \in \AI{z_j}, \label{FS:edge-cases}
\\
& 
\sum_{j=1}^{\PKT} \sum_{\trl=\FI}^{d_j} \avgpacks_j \fl j \ell {\trl} \leq \frac{\C_{\ell}}{1+\epsilon},\quad \forall \ell\in \LNKS, \label{FS:capacity}\\
& \fl j {\ell} {\trl} \geq 0,\quad \forall j \in \PD,\forall \ell \in \ELNKS,\forall \trl\in [\DMAX]_{\FI}. \quad \label{FS-postv}
\end{align}
\end{subequations}

The objective of \FS \dref{FS-obj} is to maximize the expected weighted number of arriving packets (i.e. packets with age $\trl=\FI$) that are forwarded, which corresponds to the packets attempted to be scheduled. Constraint \dref{FS:one-choice} states that the total forwarding probability of arriving packets of type $j$ out of its source node is at most $1$. If the constraint is not tight, there is a chance that the packet will not be forwarded, and therefore be dropped. 
The constraints \dref{FS:time-conserv} require that packets that have been in the system for $\tau-1$ time slots, which have been forwarded into some node $v$, must be forwarded out of the node $v$, or remain in $v$ (through a self-loop), in the following slot, with their age increasing to $\tau$. This ensures that packets are continuously forwarded, hop by hop, until their expiration\footnote{Due to the definition of route-schedules $k$ using the self-looped graph $\EG$, packets of type $j$ that are delivered earlier than their deadline to their destination $z_j$, are ``scheduled'' over $(z_j,z_j)$ until their expiry.}. Constraints \dref{FS:edge-cases} ensure that all arriving packets (i.e., with age $\trl=\FI)$ can only be forwarded out of their sources (first constraint), and must only be delivered into their destinations by their deadline (second constraint). Finally, \dref{FS:capacity} enforces the capacity constraint for each link. It requires that a link's average total traffic, due to the forwarding of packets of different ages and packet types, cannot be more than the link's capacity, divided by $(1+\epsilon)$ for some $\epsilon >0$. 
The notations and constraints of $\FS$ are illustrated through an example in \Cref{flow-route-example-2}. \textit{Note that $\FS$ is independent of the time horizon $T$, unlike $\RI$}.

\begin{figure}[t]
\centering
\begin{subfigure}{.5\textwidth}
  \centering
  \includegraphics[width=0.7\linewidth]{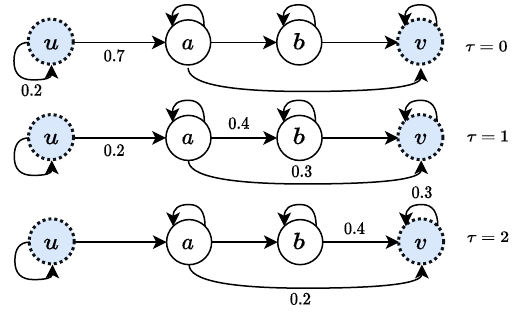}
  \caption{Forwarding variables $f_{j\ell}^\trl$, $\trl\in\{0,1,2=d_j\}$}
  \label{fig:flow-assignment-example}
\end{subfigure}%
\begin{subfigure}{.5\textwidth}
  \centering
  \includegraphics[width=0.7\linewidth]{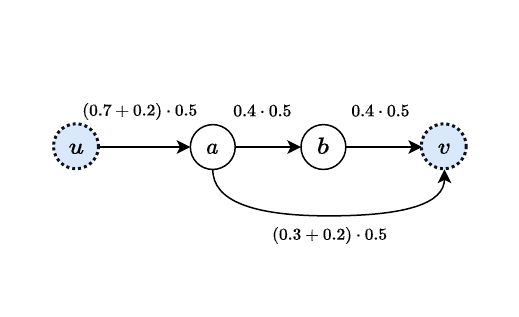}
  \caption{Links' capacity consumption}
  \label{fig:flow-constraint}
\end{subfigure}
\caption{For the network in \Cref{flow-route-example-1}, consider a packet type $j$ with arrival rate $\abar=0.5$, source $s_j=u$, destination $z_j=v$ and deadline $d_j=2$. 
\Cref{fig:flow-assignment-example} shows a feasible forwarding variables assignment $\mathbf f$ for \FS \dref{FS-obj}-\dref{FS-postv} (with, e.g., $\fl j {(u,a)} {\FI} = 0.7$ and $\fl j {(v,v)} {\TI}=0.3$). The capacity constraints \dref{FS:capacity} should sum the total flow on the link over different ages $\tau$, which is illustrated in \Cref{fig:flow-constraint}. For example, on link $(a,v)$, the total flow is $\big(\fl j {(a,v)} {\SI}+ \fl j {(a,v)} {\TI}\big) \overline a_{j} = (0.3+0.2) 0.5$.}
\label{flow-route-example-2}
\end{figure}

\RestyleAlgo{ruled}
\begin{algorithm}
 \caption{\ALGOFFSPELLED ($\ALGOFF)$}
 \label{randomized-scheduling-offline}
\SetAlgoLined
\textbf{Input:} Packet types $\{(s_j,z_j,d_j,w_j,\lambda_j)\}$, and an $\epsilon>0$.

Find optimal solution $\mathbf f^\star = \{f_{j\ell}^{t\star}\}$ to the  \LP $F_S$ \dref{FS-obj}-\dref{FS-postv}. \label{find-static}

\For{each time $t=1,2,\cdots, \TC$}{ \label{algo-time-for}
 \For{each unexpired packet of type $j$ at each node $v\in \V$}{ \label{alg-packet-for}
%
\If {age of packet is $\tau=0$} { \label{age-0-packet-if} 
Select link $\ell \in \AO v$ w.p. $ \flo j \ell {\FI}$. Else, drop w.p. $\big(1-\sum_{\ell \in \AO {s_j}} \flo j \ell {\FI}  \big)$. \label{first-link-of-packet}
}
\ElseIf{age of packet is $1\leq \tau \leq d_j$} {
Select a link $\ell \in \AO{v}$ w.p. 
$\frac{f_{j\ell}^{\trl\star}}{\sum_{\ell \in \AO{v} } f_{j\ell}^{\trl\star}}$ \label{older-packet-link}
}
Forward packet over selected link $\ell$ if there is available capacity on $\ell$ at $t$, else drop. \label{forward-if-capacity}
  } \label{alg-packet-for-end}
 }
\end{algorithm}

%
%
%
%

After solving the \LP $F_S$, we obtain a solution $\mathbf f^\star = \{\flo j \ell \tau\}$, which we use in \Cref{randomized-scheduling-offline} to schedule unexpired packets at every time slot $t$ and at every node $v$ (\Cref{alg-packet-for}-\Cref{alg-packet-for-end} in  \Cref{randomized-scheduling-offline}). Specifically, for new packets (i.e., with age $\trl=0$, \Cref{age-0-packet-if}), we select their first link $\ell \in \AO {s_j}$ with probability $ \flo j \ell {\FI}$, or drop the packet with probability
$1-\sum_{\ell \in \AO{s_j}} \flo j \ell {\FI}$ (\Cref{first-link-of-packet}). 
For packets with age $\trl>0$, we always select one of the outgoing links probabilistically (\Cref{older-packet-link}).
We then forward the packet over the selected link if there is available capacity (\Cref{forward-if-capacity}). Overall each packet is forwarded hop by hop using variables $\{\flo j \ell \tau\}$, until it reaches its destination. 
Although the algorithm only rejects packets when they arrive, once a packet is admitted, in subsequent time slots, it might also be dropped due to the capacity constraints (\Cref{forward-if-capacity}).

\Cref{main-theorem-offline} states the performance guarantee of \Cref{randomized-scheduling-offline} 
in the case that the arrivals $a_{j}^t$ for each type $j$ are i.i.d.  Bernoulli or Binomial. 

\begin{theorem}
    Given an $\epsilon \in (0,1/3)$, $\ALGOFF$ provides $(1-3\epsilon)$-approximation to $\RI$ when $\CMN \geq 2 \left(\frac{ 1+\epsilon}{\epsilon}\right)^2 \log (L/\epsilon)$ and $\TC\geq \frac{2 \DMAX^2} {\epsilon}$.  \label{main-theorem-offline}
\end{theorem}

\begin{remark}
Based on \Cref{main-theorem-offline}, we  require $\CMN\geq \Omega( \frac{\log L/\epsilon} {\epsilon^2})$ for $(1-\epsilon)$-approximation to \RI on average. This is in vast contrast to prior literature \cite{tsanikidis2022online,deng2019online} that for $\CMN=\Omega(\log ({\rho L}))$, resulted in $\Omega(1/\log ({\rho L}))$ approximation ratios for the worst case, where $\rho=\WMAX/\WMN$ was the ratio of maximum weight to minimum weight of packets. 
These findings indicate that, in situations where the traffic is stochastic, the results presented in prior work can be overly pessimistic and significant improvements can be made. Further, compared to \cite{singh2018throughput} with stochastic traffic, our results are \textit{not} asymptotic, i.e., we do not require scaling the arrival rates, capacities, and time horizon to infinity in order to obtain the stated performance guarantees. 
\end{remark}
\begin{remark}
    \Cref{randomized-scheduling-offline} admits a distributed implementation following the computation of $\FS$. Every intermediate node only requires storing the forwarding variables for each packet type based on their age.
\end{remark}
\begin{remark}
 Our results extend to unreliable networks, where the capacity of link $\ell$ at each time is random with average $\C_{\ell} p_{\ell}$ for some $p_\ell \in (0,1]$. \Cref{randomized-scheduling-offline} can be directly applied by replacing $\C_{\ell}$ with $\C_{\ell} p_{\ell}$ in $F_S$ \dref{FS-obj}-\dref{FS-postv}. In the case that the link's capacity is i.i.d. $Binomial(\C_{\ell},p_{\ell})$, we obtain a $(1-\epsilon)$-approximation ratio when $\min_{\ell }\{\C_{\ell} p_{\ell} \}= \Omega(\frac 1 {\epsilon^2}\log\frac{L} \epsilon)$.  Conceptually, the effective number of packets that can be scheduled on link $\ell$ now becomes $\C_{\ell} p_\ell$, and therefore, we require this product to satisfy a similar bound as in \Cref{main-theorem-offline}. \label{unreliable-networks-remark}
\end{remark}
\begin{remark} 
\Cref{main-theorem-offline} can be generalized to non-stationary arrival rates, dependent packet arrivals, and more general distributions than Bernoulli and Binomial. This is described in 
\Cref{generalization-non-stationary}. In particular, as a result of these extensions, in cases where a very short horizon is of interest, a non-stationary solution will still guarantee a near-optimal performance, in contrast to \Cref{randomized-scheduling-offline}.
\end{remark}

\subsection{Algorithm 2: Dynamic Learning with Probabilistic Forwarding}

In this section, we provide a dynamic variant of  \Cref{randomized-scheduling-offline} which does not require the knowledge of the arrival rates $\{\abar: j \in \PD\}$. Instead, it estimates the arrival rates dynamically over time and uses the estimates as input to \Cref{randomized-scheduling-offline}. The dynamic variant is described in \Cref{randomized-scheduling-online}.

\Cref{randomized-scheduling-online} works by dividing the time horizon $\TC$ into $(\PHASES+1)$ phases, $\phi=0,1,\cdots,\PHASES$, where $\PHASES=[\log_2(1/\epsilon)]$ for some suboptimality parameter $\epsilon \in (0,1)$. The  phases have geometrically increasing duration, with phase $\phi$ for $\phi \geq 1$ having a duration $\DMAX+\TC_\phi$, where $\TC_\phi=\TC_0 2^{\phi-1}$, and phase $\phi=0$ having duration $\TC_0 = \epsilon (\TC-\DMAX \PHASES)$. An illustration of the time division into phases is provided in \Cref{time-horizon-figure}.
During the first phase $\phi=0$, due to the lack of any statistics, the algorithm can either remain idle, or schedule packets in an arbitrary way (for example, according to algorithms in \cite{tsanikidis2022online,gu2021asymptotically}). At the end of the first phase, the algorithm estimates $\ahat$ (see \Cref{observe-algo} of \Cref{randomized-scheduling-online}) for the first time, using the number of packet arrivals of different packet types during the phase. The optimal forwarding probabilities $\mathbf {\hat f}=\{\hat f_{j\ell}^{\tau\star}\}$ are then found by solving $\FS$ (as defined in \dref{FS-obj}-\dref{FS-postv}) using the current estimates $\ahat$ instead of the actual arrival rates $\abar$ (after invoking \Cref{run-previous-algo}). The packets during the following phase are then scheduled according to the forwarding probabilities $\mathbf {\hat f}=\{\hat f_{j\ell}^{\tau\star}\}$ as in Lines \ref{alg-packet-for}-\ref{alg-packet-for-end} of \Cref{randomized-scheduling-offline} (\Cref{run-previous-algo}). More specifically, in \Cref{scaling-parameter}, we calculate a shrinkage parameter which is used to further scale down the right-hand-side of \dref{FS:capacity} (\Cref{run-previous-algo}), based on the confidence of our estimates. We allow a tuning parameter $\zeta \geq 0$ to reduce the effect of the shrinkage, with $\zeta=0$ yielding no additional shrinkage, as in  \Cref{randomized-scheduling-offline} to obtain \Cref{main-theorem-offline}. Similarly, in each new phase, we repeat the above by updating estimates $\ahat$ and therefore the forwarding probabilities.

\begin{figure}[t]
\centering
\includegraphics[width=0.6\linewidth]{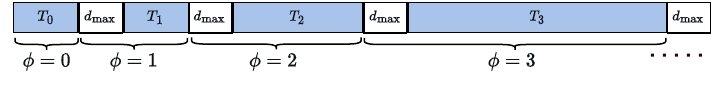}
\caption{The division of time horizon into phases of geometrically increasing duration in \Cref{randomized-scheduling-online}.
}
\label{time-horizon-figure}
\end{figure}

\RestyleAlgo{ruled}
\begin{algorithm}
 \caption{\ALGONSPELLED (\ALGON)}
 \label{randomized-scheduling-online}
\SetAlgoLined
\SetCommentSty{textit}  

\textbf{Input:} Suboptimality parameter $\epsilon \in (0,1)$. Tuning parameter $\zeta\geq 0$. 

$\PHASES \leftarrow \left[\log_2(1/\epsilon)\right]$ 
\quad \hfill {\textit{$\blacktriangleright$\  number of phases}} \label{phases-definition}

$\TC_0 \leftarrow \epsilon (\TC-\DMAX \PHASES)$  \label{t0-definition}

Run any scheduling algorithm for  $\TC_0$ slots while collecting statistics. \label{idle-algo}

\For{each time phase $\phi=1,2,\cdots,\PHASES$}{

Update estimates $\ahat$ of $\abar$ using historical samples of $\ajt$ for all $j\in \PD$ \label{observe-algo}


$\TC_\phi \leftarrow \TC_0 2^{\phi-1}$ \label{phi-phase-subduration}

$\eta \leftarrow (1-\zeta \sqrt{\log (2\PKT/\epsilon)/(\amin \TC_\phi)})$ \label{scaling-parameter}

  Run $\ALGOFF$, with estimates $\ahat$ and replacing $\C_\ell$ with $\eta \C_\ell$ in \dref{FS:capacity},  for $(\DMAX+\TC_\phi)$ time slots. \label{run-previous-algo}
}
\end{algorithm}

We remark that in \Cref{observe-algo}, estimating the arrival rates can be performed based on  a simple empirical average of the number of packet arrivals of different types using the past observations. 
The performance of \Cref{randomized-scheduling-online} is described in \Cref{online-theorem} in the case that the arrivals $\{a_{j}^t\}$ are i.i.d. Bernoulli or Binomial random variables. For sufficiently large horizon $\TC$, the algorithm is near-optimal, and only requires solving  $F_S$ a total of $[\log_2(1/\epsilon)]$ times.
\begin{theorem}
\Cref{randomized-scheduling-online} for parameters $\epsilon>0$ and $\zeta=1$, yields a $(1-10\epsilon)$-approximation to $\RI$ when $\CMN \geq 2 \left(\frac{ 1+\epsilon}{\epsilon}\right)^2 \log {\frac L {\epsilon}}$ and \[\TC=\Omega\left( \frac{a_{\max}^2}{2 \amin^2 \epsilon^2} \log(\frac{2 \PKT}{\epsilon}) + 2\DMAX^2 \log(1/\epsilon)/\epsilon \right).\]
\label{online-theorem}
\end{theorem}

\begin{remark}
Algorithm 2 provides a baseline approach under stationary model assumptions. For non-stationary packet arrival rates, a more natural approach is to update the estimates of $\ahat$ periodically, for example, by using a fixed $\TC_\phi$, or combining past estimates $\ahat$ such that more recent samples are prioritized. Further the estimates in \Cref{observe-algo} could leverage prior distributional assumptions, e.g., through a Bayesian approach, or use more complex machine learning methods.
\end{remark}

\begin{remark}
An alternative simpler two-phase algorithm, which splits the horizon into two phases $\TC_0$ and $\TC_1$ with appropriate lengths, can achieve $1-O(\epsilon)$ approximation ratio, but, for small $\epsilon$, the required time horizon will be roughly a factor $1/\epsilon$ lager than the stated $\TC$ in \Cref{online-theorem}.   

\end{remark}

\begin{remark}
Although a small $\amin$ suggests the need of a large horizon in \Cref{online-theorem}, this is a consequence of a multiplicative $(1-O(\epsilon))$ approximation ratio considered here for uniformity across our results. In practice this typically is not required as validated in our simulations in \Cref{simulation-section}. In particular, packets with very small $\amin$ effectively have a small impact on the performance of the algorithm, and can be handled as special cases if they have big weights, or are otherwise entirely ignored. This can be formalized by modifying the approximation ratio definition, such that it is satisfied up to additive constants. 
Finally, we remark that although in our analysis we assume  the knowledge of $\amin$ for simplicity, we can extend the analysis and remove this assumption by replacing $\amin$ with its estimate $\hat{\lambda}_{\min}$. In practice, $\zeta$ can either be set to $0$ or be set such that the knowledge of $\amin$ is not required.
 \end{remark}


\section{Proofs of Main Theorems} \label{analysis-technique}
In this section, we present the proof of \Cref{main-theorem-offline} and the outline for the proof of \Cref{online-theorem}.
\subsection{\texorpdfstring{Proof of \Cref{main-theorem-offline} for \Cref{randomized-scheduling-offline}}{}}\label{outline-theorem-offline} 
The proof of \Cref{main-theorem-offline} has 4 steps. In \textit{Step 1}, we construct an intermediate \LP, referred to as \EI, whose optimal objective value is higher than that of \RI. In \textit{Step 2}, we argue that this \LP admits a simple near-optimal stationary solution, thus simplifying \EI significantly into a new \LP called \EST. In \textit{Step 3}, we convert the \LP \EST into a flow-based \LP \FS, which is the \LP defined in \dref{FS-obj}-\dref{FS-postv} for \Cref{randomized-scheduling-offline}. Finally, in \textit{Step 4}, we argue that by using the forwarding variables given by \FS appropriately, we obtain \Cref{main-theorem-offline}. We describe each step below in detail.

\label{technique-static}
\textbf{Step 1. } 
We first construct a deterministic \LP below, which we refer to as the Expected Instance, \EI: 
\begin{subequations}\label{EIT}
\begin{align}
\max_{\mathbf{x}} \quad & 
 \sumtotarrv{t} w_j \abart \sumk x_{j k}^t \quad\quad\quad (:=\EI) \label{EIT:first} \\ 
\textrm{s.t.} \quad & 
\sumk x_{jk}^t \leq 1 \label{EIT:one-choice}, \quad \forall t\in [\TC],\forall j\in \PD,\\
& \sumtlast \sumpkt \abartau \sum_{k \in K_{j}:k_{t-\trl}=\ell} x_{j k}^{\trl} \leq \C_{\ell}, \quad \forall \ell\in \LNKS, \forall t\in [\TB], \label{EIT:capacity} \\
& x_{jk}^t \geq 0,\quad \forall j\in \PD, \forall t\in [\TC],\forall k \in \mathcal K_j. \label{EIT:last}
\end{align}
\end{subequations}
In \EI, the variables $\mathbf x=\{x_{jk}^t\}$ can be interpreted as probabilistic (relaxed) versions of $\{y_{jk}^{nt}\}$ (in \RI), i.e., $x_{jk}^t$ can be viewed as the probability of setting $y_{jk}^{nt}=1$ for each $n=1,\ldots,\ajt$.
Then, \EI can be viewed as an \LP for maximizing the expected reward if every packet type $j$ has $\mathbb E\big[\ajt\big]=\abart$ arrivals at each time slot $t$, and the capacity constraints are only satisfied in expectation.  

\Cref{expected-random-integer-lemma} establishes a   relationship between \RI and \EI, for general distributions on $\{\ajt\}$. 
\begin{lemma}\label{expected-random-integer-lemma}
Let $W_{\EI}$ denote the optimal value of the expected instance $\EI$, and $W_{\RI}$ denote the optimal value of $\RI$. Then, for any distribution on $\{\ajt, t \in [\TC], j\in \PD \}$, with $\abart<\infty$, we have
\[W_{\EI}\geq \E{W_{\RI}}.\]
\end{lemma}
\begin{proof}[Proof of \Cref{expected-random-integer-lemma}]
The proof is standard, and provided in \Cref{rem-proofs-offline} for completeness.
\end{proof}

\begin{remark}$W_{\EI}$ only depends on the arrival rates $\abart$ and not on the exact distribution of $\{\ajt\}$, and therefore, it serves as a universal upper bound on $\E{W_{\RI}}$ for all distributions with identical arrival rates.\end{remark}
 
We note that the expected-instance relaxations have been adopted successfully in several other problems to obtain a bound on an integer program  \cite{jiang2020online,sun2020near,devanur2019near,gallego1994optimal}. Here, we aim to construct a scheduling algorithm that obtains a total reward close to the optimal value of \EI, which implies an approximation to \RI due to \Cref{expected-random-integer-lemma}. 
However, we cannot directly solve \EI, as this LP has a number of constraints which grow with the time horizon $\TC$, and moreover, the number of variables could be exponential, due to the exponentially large number of route schedules. We address these issues in Step 2 and Step 3 below.

\textbf{Step 2.} To resolve the issue of the number of capacity constraints \dref{EIT:capacity} in
$\EI$ scaling with $\TC$, we show that in the stationary case of $\abart\equiv \abar$, a time-independent (stationary) solution is near-optimal for \EI, thus allowing us to simplify \EI. We state the result in \Cref{static-is-good-lemma} below. 

\begin{lemma} \label{static-is-good-lemma}
In the stationary case, i.e., $\abart\equiv \abar$,  there is a stationary solution $x_{ik}^t=x_{ik}^\star$ for $\EI$ with objective value $\overline W_{\EI}$ that satisfies 
\begin{equation}
W_{\EI}  \geq \overline W_{\EI} \geq W_{\EI} \left(1- 2 \frac {\DMAX^2}{\TC}\right). \label{static-is-good-inequality}
\end{equation}
\end{lemma}
\begin{proof}[Proof of \Cref{static-is-good-lemma}]
The first inequality in \dref{static-is-good-inequality} follows directly since limiting our solutions to static solutions is an additional constraint on the maximization, and therefore it can only result in a lower objective value. 

We prove the second inequality in \dref{static-is-good-inequality} in three steps. We first provide an overview of the steps. In Step A, we derive from $\EI$, a new optimization problem $\EIB$, which is identical in most of the capacity constraints, except for the capacity constraints for times $t<\DMAX$ or $t>\TC$ which are modified such that they contain more terms (therefore $\EIB$ is more constrained, and any solution to it, will also be feasible for $\EI$). More specifically, these constraints are modified appropriately such that $\EIB$ is symmetrical with respect to time shifts of the solution. In Step B, we leverage the  symmetry of $\EIB$ to argue that it admits a stationary optimal solution. In Step C, we argue that the optimal value for $\EIB$ is close to that of $\EI$. Since any feasible solution to $\EIB$ is also feasible for \EI, and $\EIB$ admits a stationary optimal solution with objective value close to the optimal of $\EI$, then there is a near-optimal static solution to \EI (as in the statement of \Cref{static-is-good-lemma}). 

\textbf{Step A}. Note that $\EI$ is nearly symmetrical w.r.t. time shifts. With a modification to this optimization program, we arrive to a symmetrical optimization program $\EIB$. Consider the capacity constraints \dref{EIT:capacity} in \EI.
If $t<\DMAX$ or $t>\TC$, due to the lower or upper limit of the outer summation being clipped to $0$ or $\TC$, the symmetry of the problem is broken. We can make the program cyclically symmetrical, by appropriately extending the outer summation when it is clipped, and using the modulo operator. Specifically, consider the modified capacity constraints for $\EIB$:
\begin{equation}
     \sum_{\trl=(t-\DMAX)}^{t}  \sumpkt \abar \sumschedl x_{j k}^{\trl \ \mathsf{mod} \ (\TC +1)} \leq \C_{\ell}, \forall \ell,\forall t\in \TC. \label{cap-constr-modified}
\end{equation}
Note that this modification only impacts the constraints \dref{EIT:capacity} for $t<\DMAX$ or $t>\TC$ in \EI. These constraints are replaced with the stronger constraints \dref{cap-constr-modified} for $t\in[\DMAX]$ in \EIB. This modification is visualized in \Cref{fig:illustrate-symmetrizing}. 

\textbf{Step B}. Here, we prove the existence of a stationary optimal solution in \EIB. To see this, consider any optimal solution to $\EIB$ (which might be time-dependent), denoted with $\mathbf x^\star=\{x_{jk}^{t\star}\}$.
Let us denote a time-shift of all variables of $\mathbf x^\star$ with $\tprm$ slots by $\mathbf x[\tprm]$, i.e., $x_{lk}^{t}[\tprm] :=x_{jk}^{((t+\tprm)\textsf{mod} \TC)\star}$. Note that $\mathbf x[0]=\mathbf x^\star$. As the program is time-shift invariant, $\mathbf x[\tprm]$ should also be optimal for all $\tprm$. Note that there are at most $\TC$ distinct shifts for $\mathbf x[\cdot]$ which are optimal solutions.
Since $\EIB$ is an \LP, the average of a set of optimal solutions yields a feasible solution that is also optimal. Hence,  
$\overline{\mathbf x}^\star := \frac{1} \TC \sum_{\tprm=1}^\TC \mathbf x[\tprm]$ is an optimal solution. 
This solution is stationary, as it can easily be verified that it corresponds to the time-average of the original variables, i.e.,
$\overline x_{jk}^\star = \frac{\sum_{t=1}^\TC x_{jk}^{t}}{\TC}.
$

\textbf{Step C.}
Finally, we finish the proof of \Cref{static-is-good-lemma}. Note that  any feasible solution of $\EI$ can be converted to a feasible solution for \EIB by setting all variables $\vEI t$ for times $t<\DMAX$ or $t>\TC-\DMAX$ to $0$. Indeed, this results in a feasible solution for $\EIB$ since the only constraints that are different between $\EIB$ and $\EI$, only involve the variables at these boundary time-slots. Therefore we need to set the variables to $0$ for a range of $2\DMAX$ time slots. The maximum loss in reward, however, due to setting these variables to $0$, is, per time slot, at most $W_{\EI} \DMAX/\TC$. Indeed, consider for example time slot $0$. If the total reward could be higher than $W_{\EI} \DMAX/\TC$, then scheduling every $\DMAX$ time slots only, we could obtain a reward higher than $W_{\EI} \DMAX/\TC \times  \TC/\DMAX=W_{\EI}$, which leads to a contradiction.  

Since at every time slot we can lose a maximum reward of $W_{\EI} \DMAX/\TC$, in $2\DMAX$ time slots we lose a maximum reward of at most $W_\EI 2 \frac{\DMAX^2} {\TC}$. 
\begin{figure}[t]
\centering
\begin{subfigure}{.40\textwidth}
  \centering
  \includegraphics[width=0.95\linewidth]{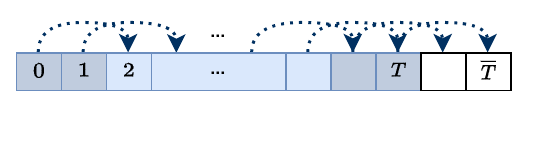}
  \caption{Range of capacity constraints for \EI}
  \label{fig:before-symmetrize}
\end{subfigure}%
\begin{subfigure}{.40\textwidth}
  \centering
  \includegraphics[width=0.95\linewidth]{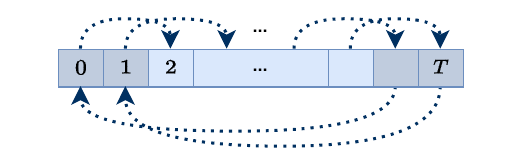}
  \caption{Modified range for \EIB}
  \label{fig:after-symmetrize}
\end{subfigure}
\caption{\Cref{fig:before-symmetrize} shows the time horizon for \EI when $\DMAX=2$. The starting point of each arrow indicates the earliest time slot $\trl$ with arrivals  that impact time slot $t$, which is indicated by the end of an arrow, as seen in \Cref{EIT:capacity}. For example time slot $t=2$ can have arrivals as early as $\tau=0$. The symmetry of the problem across time slots breaks at $t=0,1,\TC+1,\TB$ (shaded).  Adding terms $x_{jk}^\tau$ for these time slots in \dref{EIT:capacity} allows us to create a set of stricter constraints \dref{cap-constr-modified} that make \EI symmetrical. This is shown in  \Cref{fig:after-symmetrize}), where, e.g., now at the capacity constraint for $t=0$, we have terms $x_{jk}^\tau$ from $\trl=\TC+1,\TB$. }
\label{fig:illustrate-symmetrizing}
\vspace{- 0.1 in}
\end{figure}
\end{proof}

Due to \Cref{static-is-good-lemma}, we can search for stationary solutions, which simplifies \EI significantly. We refer to the simplified \LP  as the \textit{stationary expected instance program} $\EST$, defined in \dref{ES-obj}-\dref{ES-postv} below:
\begin{subequations}\label{ES}
\begin{align}
\max_{\mathbf{x}} \quad & 
\sumpkt w_{j} \abar \sumk x_{jk} \quad(:=\EST) \label{ES-obj} \\ 
\textrm{s.t.} \quad & 
\sumsched x_{jk} \leq 1,  \quad \forall j\in \PD, \label{ES:one-choice}\\
& 
\sumpkt\sum_{\trl=0}^{d_j}  \abar \sum_{k \in \mathcal K_{j}:k_{\trl}=l} x_{j k} \leq \C_{\ell}, \quad \forall \ell\in \LNKS, \label{ES-capacity}
\\ 
& x_{jk} \geq 0,\quad \forall j\in \PD,\forall k \in \mathcal K_j.  \label{ES-postv}
\end{align}
\end{subequations}

Note that compared to $\EI$, the capacity constraints \dref{ES-capacity} are much simpler, as constraints \dref{EIT:capacity} in $\EI$ for $t \in [\DMAX , \TC]$ become identical, and the constraints for $t<\DMAX$ or $t>\TC$ are weaker. 
Therefore, 
we only need to keep the constraints \dref{EIT:capacity} for only one $t \in [\DMAX , \TC]$, e.g., for $t=\DMAX$. 

    \textbf{Step 3.} 
     The simplified \LP $\EST$ still has exponentially many variables (there are exponentially many route-schedules $k$). To resolve this, we work with the \textit{flow-based reformulation} of \EST through the \textit{forwarding variables} defined in Section~\ref{sec:algo-results} for \Cref{randomized-scheduling-offline}. Analogous reformulations have been utilized to a great extent in network flow problems \cite{magnanti1993network} (for more details, refer to \Cref{flow-relation}). The conversion results in the flow-based \LP $\FS$, defined earlier in \dref{FS-obj}-\dref{FS-postv}, but with $\epsilon =0$.
To see that, recall that $x_{jk}$ is the probability of scheduling packet-type $j$ over route-schedule $k$, and forwarding variable  $\fl j \ell \tau$ can be interpreted as the fraction of type-$j$ packets scheduled over link $\ell$ at age $\tau$. 
    Then, \dref{ES-obj} and \dref{ES:one-choice} are equivalent to \dref{FS-obj} and \dref{FS:one-choice}, respectively, since $\sum_{k\in K_j} x_{jk} \equiv \sum_{\ell\in \AO{s_j}} \fl j {\ell} {\FI}$, where, recall that $\AO{s_{j}}$ denotes the outgoing links of $s_j$ in $\EG$. Similarly, capacity constraint \dref{ES-capacity} is equivalent to \dref{FS:capacity} with $\epsilon =0$.     
    The remaining constraints \dref{FS:time-conserv}, \dref{FS:edge-cases}, \dref{FS-postv} are the flow conservation laws.        
    At the end, similarly to network flow problems, we can map the solution $\{\fl j \ell \trl\}$ to a solution for $\{x_{jk}\}$ due to the flow decomposition theorem \cite{goldberg1998beyond}.  
Note that in $\FS$, we have reduced the capacity in constraint \dref{FS:capacity} by a factor $(1+\epsilon)$. It is easy to verify that solving \FS, with $\epsilon >0$, results in at most a factor $\frac{1}{1+\epsilon}$ reduction in the optimal objective value, and consequently, $W_{\FS}\geq \frac{W_{\EST}}{1+\epsilon}$.

\textbf{Step 4.} In this step, we tie our analysis together as follows.  If we could schedule the packets successfully according to forwarding probabilities $\{f_{j\ell}^{\tau \star}\}$ we would obtain total expected reward close to $\EI$ (\Cref{static-is-good-lemma} and Step 3). However, scheduling according to these probabilities might violate the capacity constraints of $\RI$. 
We argue that due to the gap between the expected traffic and the actual capacity on a link, imposed in \dref{FS:capacity} by $\epsilon >0$, the probability that traffic surpasses any link's capacity is significantly reduced. In fact, we obtain an ``almost feasible'' solution to $\RI$, if the capacity is greater than a certain threshold. Then, by dropping packets attempted to be scheduled on a link with full capacity, we obtain a feasible solution to $\RI$. Further, the drop probability of any packet is small (i.e. bounded by $\epsilon$) as each link over which the packet might be transmitted is unlikely to be at full capacity. 
The above result is stated formally in \Cref{small-drop-chance-lemma}.

\begin{lemma}\label{small-drop-chance-lemma}
Assume, for each $j \in \PD$, $\{\ajt, t\geq 0 \}$ are i.i.d. Bernoulli (or Binomial) random variables. Then using forwarding variables  $\{f_{j\ell}^{\tau \star}\}$ for scheduling packets (Lines \ref{alg-packet-for}-\ref{alg-packet-for-end} of \Cref{randomized-scheduling-offline}), the probability of a packet being dropped is at most $\epsilon$, if $\CMN \geq 2 \left(\frac{ 1+\epsilon}{\epsilon}\right)^2 \log \frac{L}{\epsilon}$.
\end{lemma}
\begin{proof}[Proof of \Cref{small-drop-chance-lemma}]
In the proof, we leverage the following concentration inequality from \cite{chung2006complex}.

\begin{lemma}\label{concentration-lemma}
Suppose $X_1,X_2,\cdots,X_n$ are independent, but not necessarily identically distributed random variables, with $X_i \leq B$, for all $i \in [n]$. Let $X := \sum_{i=1}^{n} X_i$ and
$S:=\sum_{i=1}^n \E{X_i^2}$. Then
\begin{equation}
\Pr{[X \geq \E{X}+\lambda]} \leq e^{-\frac{\lambda^2}{2 S+ 2 B \lambda/3}}. \label{conc-inequ-inequality}
\end{equation}
\end{lemma}
Using \Cref{concentration-lemma}, we prove that exceeding the capacity of any link happens with small probability, which will then allow us to bound the probability of a packet being dropped. First, in accordance with the notation in \Cref{concentration-lemma}, we define the variables:
\[
X_{j\trl}=\mathds 1 (\text{a type-$j$ packet  arrives at time $\trl$ and is scheduled at $\ell$ at time $t$})
\]
where $\mathds 1(\cdot)$ is the indicator function.
The total capacity consumption on $\ell$ at time $t$ can be written as:
\begin{equation} 
X:=\sum_{j=1}^{\PKT}\sum_{\trl=t-\DMAX}^t X_{j\trl}, \label{total-consumption-random-variable}
\end{equation}
where we assume the case of $\DMAX \leq t\leq \TC$ for simplicity, but the same holds for all $t$ (since we have fewer terms in the other two cases, the probability of packet drops will be lower).
Note that
\begin{equation}
\E{(X_{j\trl})^2}=\E{X_{j\trl}}=\abar
f_{j\ell}^{(t-\trl)\star}, \label{total-consumption-expected-value}
\end{equation}
where the first equality is due to  $X_{j\trl}\in\{0,1\}$ and the second equality follows by a simple calculation,
using the probability of a type-$j$ packet arriving at $\trl$ and being scheduled with total probability $f_{j\ell}^{(t-\trl)\star}$ over $\ell$, where the 
optimal forwarding probabilities $\mathbf f^\star $ found through \FS are used.
Taking an expectation on the expression in \dref{total-consumption-random-variable} and utilizing \dref{total-consumption-expected-value}, and using the fact that variables $\mathbf f$ are a solution to $\FS$ and thus they satisfy \dref{FS:capacity},  we have
$S = \E{X} \leq \frac{\C_\ell}{1+\epsilon}:=S_u$. 
Then, due to the i.i.d. Bernoulli (or Binomial) assumption, all random variables in the right-hand-side of \dref{total-consumption-random-variable} are independent, and therefore we can apply \Cref{concentration-lemma}, with $B=1$, for $\lambda=\epsilon S_u = \frac{\epsilon \C_\ell}{1+\epsilon}$, and using $S\leq S_u$. For the exponent in \dref{conc-inequ-inequality} we have:

\[
-\frac{\lambda^2}{2S+2 \lambda/3}=- \frac{\epsilon^2 S_u^2 }{2S_u+2 \epsilon S_u /3 } \leq -\frac { \epsilon^2 S_u }{2 + 2 \epsilon/3} \leq -
\frac { \epsilon^2 S_u }{2(1+\epsilon)}  =- \frac {\epsilon^2 \C_\ell }{2(1+\epsilon)^2}
\]
Therefore using inequality \dref{conc-inequ-inequality} we get:
\[
\Pr {[X \geq \frac{\C_\ell}{1+\epsilon} + \frac{\epsilon \C_\ell}{1+\epsilon} ]} =
\Pr {[X \geq \C_\ell ]} \leq e^{- \frac{\C_\ell} {2} (\frac{\epsilon}{1+\epsilon} )^2 } \overset{(a)}\leq  
e^{\log \frac \epsilon L } = \epsilon/L 
\]
where in $(a)$ we have used the stated minimum link capacity in \Cref{small-drop-chance-lemma}.
Now, consider an arbitrary packet which is attempted to be scheduled over a route-schedule with a maximum of $L$ link-time slot pairs. Then, the probability that any link in the route-schedule exceeds the capacity constraint can be bounded through a union bound that gives the statement of the Lemma, i.e., the probability of dropping the packet is bounded by $\epsilon$.\end{proof}

We note that the proof extends to non-stationary arrival rates, and can also be generalized to dependent packet arrivals (\Cref{small-drop-chance-lemma-generalization}).
Since packets are dropped with probability at most $\epsilon$, the performance loss of the algorithm is at most a multiplicative factor $\epsilon$ (in addition to factor $1/(1+\epsilon)$ due to the scaling of the capacities in \dref{FS:capacity}, as discussed in Step 3).
Overall, the expected obtained reward is close to $\EI$, which, by \Cref{expected-random-integer-lemma}, gives a guarantee relative to the average optimal value of $\RI$, yielding \Cref{main-theorem-offline}. This is detailed in the following proof.

\begin{proof}[Proof of \Cref{main-theorem-offline}]
The expected reward in a time slot is:
\begin{align*}
&\E{W_{\ALG}} =
\E{
\sumtotarrv{t} w_j \sumajt \sumsched \vRI {nt}
} = 
\sumtotarrv{t} w_j \E{\sumajt \sumsched \vRI {nt}} \overset{(a)}=
\\
& \sumtotarrv{t} w_j \abart \E{\sumsched\vRI {nt}} \overset{(b)} {\geq}
\sumtotarrv{t} w_j \abart \sum_{\ell\in \mathcal O(s_j)}  (1-\epsilon) \fl j {\ell} {\FI\star} \overset{(c)}=
\frac{1-\epsilon}{1+\epsilon} \overline W_{\EI} \overset{(d)}\geq \\
&
\frac{1-\epsilon}{1+\epsilon} \left(1- 2 \frac {\DMAX^2}{T}\right)W_{\EI}  \overset{(e)}\geq 
  \frac{1-\epsilon}{1+\epsilon} \left(1- 2 \frac {\DMAX^2}{T}\right)\E{W_{\RI}} \overset{(f)}\geq (1-3 \epsilon) \E{W_{\RI}},
\end{align*}
where $(a)$ follows due to the independence between $\{\ajt\}$ and decisions $\{\vRI {nt}\}$, $(b)$ follows by Lemma~\ref{small-drop-chance-lemma} due to the forwarding probabilities under \ALGOFF, $(c)$ follows due to Step 3, $(d),(e)$ follow due to Lemmas~\ref{static-is-good-lemma},\ref{expected-random-integer-lemma} respectively, and $(f)$ follows for $\TC\geq \frac{2 \DMAX^2} {\epsilon}$. 
Note that the dependence on an i.i.d Bernoulli or Binomial distribution only appears on $(b)$ and can be lifted by extending Lemma~\ref{small-drop-chance-lemma}.
\end{proof}

\subsection{\texorpdfstring{Proof Outline of \Cref{online-theorem} for \Cref{randomized-scheduling-online}}{}}

Here we present the outline for the proof of \Cref{online-theorem} for \Cref{randomized-scheduling-online}, which is divided into 3 steps. The detailed proofs are provided in \Cref{proofs-online}.

\textbf{Step 1.} In the first step, we find the number of samples required to obtain estimates for $\{\abar\}$ within a desired multiplicative accuracy $(1\pm \epsilon_0)$, for some parameter $\epsilon_0$. This is stated in \Cref{number-of-samples}.

\begin{lemma}\label{number-of-samples}
For any i.i.d. distribution on $\{\ajt\}$, to estimate $\{\abar\}$ such that $\abar(1-\epsilon_0) \leq \ahat \leq \abar (1+\epsilon_0), \forall j$,
with probability at least $1-\epsilon$, given $\epsilon,\epsilon_0 \in (0,1)$, we require samples of the arrivals from $\frac{a_{\max}^2}{2 \amin^2  \epsilon_0^2} \log \left ( \frac {2 \PKT}{\epsilon}\right)$ time slots.
\end{lemma}

\textbf{Step 2.} Towards the goal of solving $\EI$ near-optimally as in Algorithm 1,  and due to the results outlined in \Cref{technique-static}, we focus on solving $\FS$ near optimally. As discussed in \Cref{observe-algo} and \Cref{run-previous-algo} (\Cref{randomized-scheduling-online}), we substitute $\abar$ with estimates $\ahat$ of $\abar$ in $\FS$, to obtain a solution $\hat {\mathbf f}$. 
Subsequently, by appropriately scaling this solution, we obtain a near-optimal and feasible solution to $\FS$ with high probability. This is formalized in \Cref{static-approximation}.
\begin{lemma} \label{static-approximation}
Let $\hat {\mathbf f}$ be the solution to $\FS$ by using estimates $\ahat$, satisfying $\abar (1-\epsilon_0) \leq \ahat \leq \abar(1+\epsilon_0)$, in $\FS$. Let
$\tilde {\mathbf f}=(1-\epsilon_0) \hat {\mathbf f}$. The scaled solution $\tilde {\mathbf f}$, whose objective value we denote with $\tilde F_S$, is, with probability  at least $1-\epsilon$, a feasible solution to $\FS$ and
\[
\tilde F_S \geq (1-4\epsilon_0) \FS.
\]
\end{lemma}

\textbf{Step 3.} Based on \Cref{static-approximation}, we can solve $F_S$ with progressively better estimates $\ahat$ in $\log(1/\epsilon)$ phases, to obtain an increasingly better approximation in each phase for $\RI$, following the connection between $\FS$ and $\RI$ discussed in \Cref{technique-static}. More specifically, in each phase we have a different number of samples available and based on \Cref{number-of-samples} and \Cref{static-approximation}, $\epsilon_0$ is decreasing across phases (and determined by the carefully designed phase durations). Aggregating the approximations of each phase, we obtain an overall near-optimal solution assuming sufficient time is available and under the capacity constraints in the statement of \Cref{randomized-scheduling-online}. Refer to \Cref{proofs-online} for the detailed proofs.

%



\section{Discussion} \label{discussion-section}

In this section, we discuss how our technique can be extended to other distributional assumptions and discuss some complexity and implementation improvements.

\subsection{Extension to more general arrival distributions}\label{generalization-non-stationary}
\textbf{General stationary distributions.} The theoretical results in \Cref{main-theorem-offline} and \Cref{online-theorem} extend to general distributions on the arrivals, other than i.i.d. Binomial and Bernoulli arrivals. The distribution affects the required minimum link-capacity, and generally, improved results are obtained when packet arrivals are less dependent on each other. Specifically, if the packet arrivals of a packet-type in any $(\DMAX+1)$ consecutive time slots can be partitioned into a set of groups of maximum size $D$ such that each group's arrivals are independent of the arrivals of other groups, then our results can be generalized, requiring capacity $\CMN \geq 2 D \left(\frac{ 1+\epsilon}{\epsilon}\right)^2 \log (L/\epsilon)$ for near-optimal ($1-3\epsilon$)-approximation (\Cref{main-theorem-offline-dependencies} in \Cref{dependent-packet-arrivals}). Binomial and Bernoulli arrivals correspond to $D=1$.

We remark that this extension allows a dependence between packets within the same time slot or across different time slots. For example, if the arrivals of a packet type are independent across time, and are at most $D$ at any time, then the above partitioning is possible. Alternatively, packets might depend across time. For example, when $a_{\max}=1$, with the arrivals of a packet type at any time, dependent on the arrivals of that type in prior slots, we obtain near-optimal performance for 
$\CMN \geq 2 (\DMAX+1) \left(\frac{ 1+\epsilon}{\epsilon}\right)^2 \log (L/\epsilon)$, that is, $D=(\DMAX+1)$ in this case. 

The formal description of these results is provided in \Cref{dependent-packet-arrivals}.

\textbf{Known non-stationary traffic.} Our assumption on identically distributed arrivals across time slots can be relaxed. In that case, due to the presence of non-stationarity, we cannot rely on stationary forwarding probabilities $f_{j\ell}^\trl$ (identified earlier in \FS) to obtain a near-optimal solution. As a result, in contrast to the i.i.d. case, the forwarding variables will need to be time dependent, i.e., $f_{j\ell}^{\tau t}$, in addition to the dependence on link, type and age of the packet. 
The modified \LP is referred to as \FNS, and is presented in \Cref{non-stationary-extension}, see \dref{FS-obj-TIME}-\dref{FS-postv-TIME}. For example, the objective in $\FS$ should be modified to describe the average total reward (rather than the average per time slot reward)
and it becomes:
$$
\sum_{t=1}^{\TC} \sum_{j=1}^{\PKT} w_{j} \abart \sum_{\ell \in \AO{s_{j}}} \fl{j} \ell {\FI t}.
$$
Further, the link capacity constraints \dref{FS:capacity} should be adjusted such that they are applied for each time slot independently (as opposed to a single constraint that captures the average traffic on the link).  Note that the modification for non-stationarity comes with a computational cost, as the number of variables and constraints now scale also with the horizon $\TC$. The cost of solving \FNS with a large horizon can be mitigated by  splitting the horizon into frames, and solving \FNS in each independent frame. A trade off between complexity and performance arises, with larger considered frames generally yielding better performance but with higher complexity  (we provide more details in \Cref{frame-base-extension-appendix}, e.g., \Cref{frm-constr-lemma}).
Following the determination of forwarding variables $f_{j\ell}^{\trl t}$, we can proceed by probabilistically forwarding any packet of type $j$ arriving at time $t$ according to variables $f_{j\ell}^{\trl t}$ (similarly to \Cref{randomized-scheduling-offline}). We provide the general  \Cref{most-general-thm} in \Cref{non-stationary-extension}.

\textbf{Periodic traffic distribution.} In the special case of periodic traffic distribution, the complexity of the non-stationary approach can be further reduced such that the number of variables and constraints do not scale with $\TC$, but rather, scale  with the length of the period of the periodic traffic. In this case, our proof techniques outlined in \Cref{analysis-technique} (such as \Cref{static-is-good-lemma}) can be modified to show that a near optimal solution is periodic (as opposed to stationary). For example, for a period of length $\widetilde {\TC}$, the forwarding 
variables are 
$\{f_{j\ell}^{\trl t},t\in [\widetilde {\TC}]\}$ and a packet arriving at time $t^\prime$ must be forwarded according to variables $\{f_{j\ell}^{\trl (t^\prime\mathrm{mod}\ \widetilde {\TC})}\}$ . For more details, refer to  \Cref{periodic-traffic-extension}.

\textbf{Unknown arrival rates.} For the non-stationary traffic when the model is not known, our techniques can be used in conjunction with prediction methods that estimate future arrival rates $\abart$. In that case, the predicted future $\abart$ is used as input to $\FNS$. In particular, the performance of our methods may be characterized based on the accuracy of the predictions. 
Prediction-based methods have been studied recently in   \cite{stein2023learning}. By leveraging an extension to \Cref{static-approximation}, we can obtain performance which depends on the accuracy of the estimated arrival rates $\widehat{\abart}$. Further, by using frame-based methods, such as the one proposed in \Cref{frm-constr-lemma} (\Cref{frame-base-extension-appendix}), we only require predicting the arrival rates for an incoming frame.

\subsection{Complexity and implementation}\label{complexity-section}

\textbf{Reducing the number of used route-schedules.} In our algorithms, we described a hop-by-hop randomization, in order to convert the forwarding probabilities to a particular route-schedule. However, depending on the design needs, other conversions might be preferred and are available. For example, an iterative procedure based on \cite{raghavan1987randomized} can be used to find a fixed set of route-schedules, over which the source of the packet can randomize over for scheduling the packet. Specifically, each packet type will have a maximum of $\DMAX |\ELNKS|$ route-schedules over which the packets of this type are scheduled. This is in contrast to the hop-by-hop approach that can potentially choose any route-schedule. In certain scenarios, the former method might be preferred. 
More details on this alternate process are provided in \Cref{iterative-process-flow-to-route}.

\textbf{Solving the LP {\FS}.} Recall that \FS \dref{FS} is an LP which has $O(\DMAX |\ELNKS| \PKT)$ number of variables and constraints and therefore can be solved efficiently in polynomial time.
Following our results, it is interesting to seek even more efficient implementations of \FS or online solutions of \FS, through, for example distributed primal-dual techniques, or multicommodity-flow-inspired solutions \cite{magnanti1993network}.  

\section{Simulation Results} \label{simulation-section}

In this section, we report simulation results using both synthetic data and real network data.

\subsection{Evaluation using synthetic traffic} \label{synthetic-data}
\textbf{Comparisons between algorithms.} First, we consider network $\G_1$ shown in \Cref{fig:networka-competitive}. In our first experiment, we examine traffic generated using $10$ packet types, with source $s_j$ and destination $z_j$ for each type $j$ chosen randomly in the network, weights $w_j$ chosen randomly in $(0,1)$, and arrival rates $\abar$ uniformly chosen in $(0,25)$. The arriving packets have a deadline of $10$, and the time horizon is $\TC=5000$.

\begin{figure}[t]
\centering
\begin{subfigure}{.35\textwidth}
  \centering
  \includegraphics[width=0.95\linewidth]{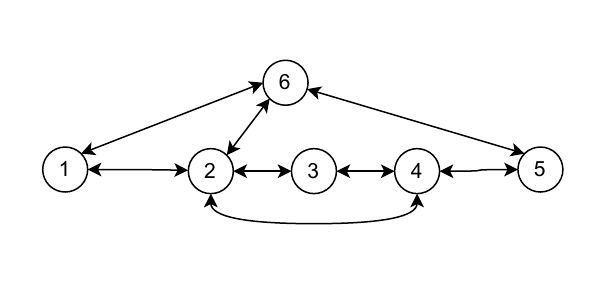}
  \caption{Simple network $\G_1$}
  \label{fig:networka-competitive}
\end{subfigure}%
\begin{subfigure}{.3\textwidth}
  \centering
  \includegraphics[width=0.95\linewidth]{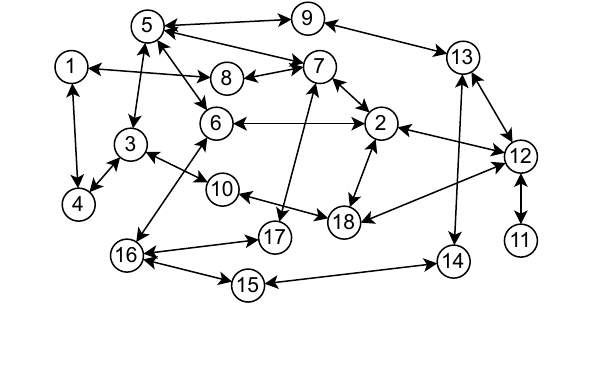}
  \caption{IBM network $\G_{IBM}$}
  \label{fig:network-ibm}
\end{subfigure}
\caption{Two different networks used in simulations.}
\label{fig:guarantees-comparison}
\vspace{- 0.1 in}
\end{figure}

We evaluate three variants of our algorithms: \ALGOFF (\Cref{randomized-scheduling-offline}), and two variants of \ALGON (\Cref{randomized-scheduling-online}), the base method which divides the horizon into exponentially increasing phases (denoted with \ALGON-EXP), and a method that divides the horizon into phases of fixed duration of $500$ slots (\ALGON-500). 
With the goal of verifying the near-optimal behavior of our algorithms as $\CMN$ increases, we additionally provide comparisons with the recent worst-case-based algorithm GLS-FP \cite{gu2021asymptotically} (Fastest-Path variant of GLS), which schedules packets greedily on the ``fastest path`` (the route-schedule that delivers the packet as early as possible) from source to destination, when enough capacity is available, otherwise it drops the packet. This is the algorithm of choice in the simulations of \cite{gu2021asymptotically}, which in the simulations therein outperforms works such as the one in \cite{deng2019online}. 

To evaluate our algorithms with respect to the optimal, we used an upper bound on the optimal value of \RI based on the optimal value of \FS (\dref{FS-obj}-\dref{FS-postv}) (multiplied by $(1-2 \DMAX^2/\TC)^{-1}$ due to \dref{static-is-good-inequality} and our earlier analysis, outlined in \Cref{outline-theorem-offline}). Hence, we can obtain a lower bound on the approximation ratio of each algorithm, which would be the metric in our comparisons here.

In our first experiment, we study the impact of link capacity with fixed traffic.  We considered an equal link capacity  $\C_\ell=\CMN$ for all the links in the network and varied $\CMN$ between $1$ and $25$. The results are shown in \Cref{fig:comparison-a-synthetic}. We observe that all our algorithms outperform GLS-FP significantly, and even for very small capacities, obtain more than a fraction $0.8$ of the optimal value. 

In the next experiment, we study the impact of increasing capacity $\CMN$ while also scaling the traffic intensity proportionally (to avoid simplification of the problem). We do so by increasing the number of packet types with the capacity, with $\PKT=\{5,10,\cdots,40\}$ and $\CMN=\{1,4,\cdots,25\}$ (as earlier). The results are depicted in \Cref{fig:comparison-b-synthetic}. We observe that our algorithms improve the performance with the scaling of capacity and traffic regardless, whereas GLS-FP does not improve as much as in \Cref{fig:comparison-a-synthetic}. This is in line with our analysis that suggests an improvement of the approximation ratio for larger capacities.

The good performance of our algorithms was further validated in an extensive set of topologies, deadlines and traffic configurations. Refer to \Cref{additional-simulations} for additional simulation results.

\textbf{Performance over different arrival distributions.} As discussed in \Cref{generalization-non-stationary} (with more details in \Cref{dependent-packet-arrivals}), our methods extend to general packet arrival distributions, with the characteristics of the distribution only affecting the required $\CMN$ to obtain near-optimal performance. Here, we validate the good performance of \ALGOFF, under three different distributions, for identical arrival rates in each case. In particular, we consider an i.i.d. process with arrivals at each time slot and for each packet type from the standard Binomial distribution, a Poisson distribution, and a scaled Bernoulli distribution which is either $0$ or $a_{j}$ for packet type $j$ (all with equal $\lambda_j$). In particular, we note  that i.i.d. arrivals from the Binomial distribution result in the minimum dependency degree $D=1$, whereas those from Scaled Bernoulli have maximum dependency degree, with $D=a_{\max}$. We use
IBM's network $\G_{IBM}$ \cite{ibmnetwork}, shown in \Cref{fig:comparison-a} (similar results are obtained for $\G_1$), which consists of $18$ nodes and $20$ links. Here, we consider higher arrival rates compared to the previous simulations, randomly selected in $(0,150)$ for each $j$, as we seek to study the performance for larger capacities in order to verify that approximation ratios are improved over all distributions as $\CMN$ increases. The approximation ratio bounds are shown in \Cref{fig:comparison-diff-distrs}. We notice that the best performance is obtained under the Binomial distribution, in line with our theoretical results ($D=1$, \Cref{main-theorem-offline-dependencies}). The Poisson distribution also exhibits good performance. The worst performance is obtained (as expected) for the Scaled Bernoulli $(D=a_{\max})$.


\subsection{Evaluation using real traffic} \label{sim-real-traffic}

\begin{figure}[t]
\centering
\begin{subfigure}{.33\textwidth}
  \centering
  \includegraphics[width=0.9\linewidth]{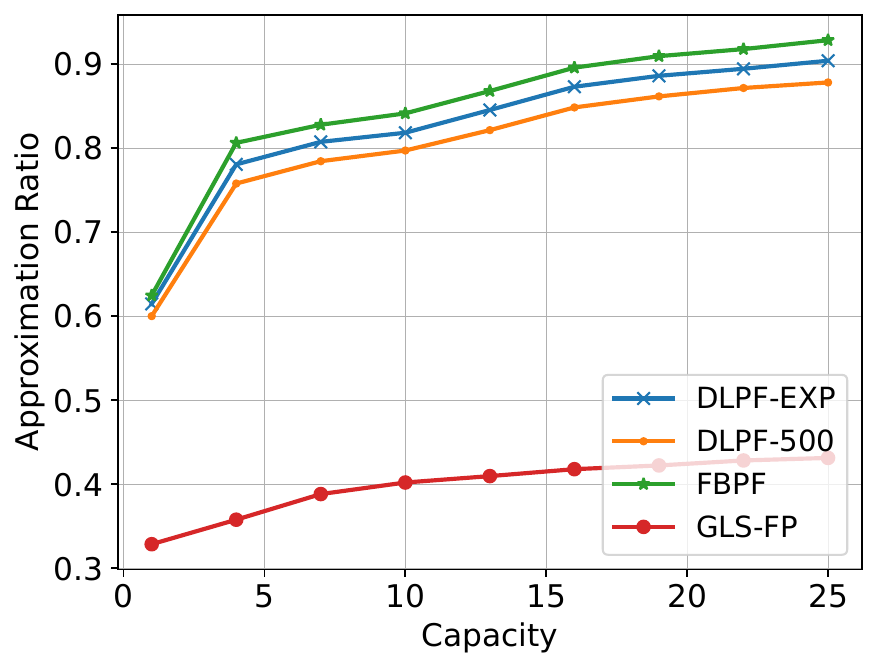}
  \caption{Varying capacity}
  \label{fig:comparison-a-synthetic}
\end{subfigure}%
\begin{subfigure}{.33\textwidth}
  \centering
  \includegraphics[width=0.9\linewidth]{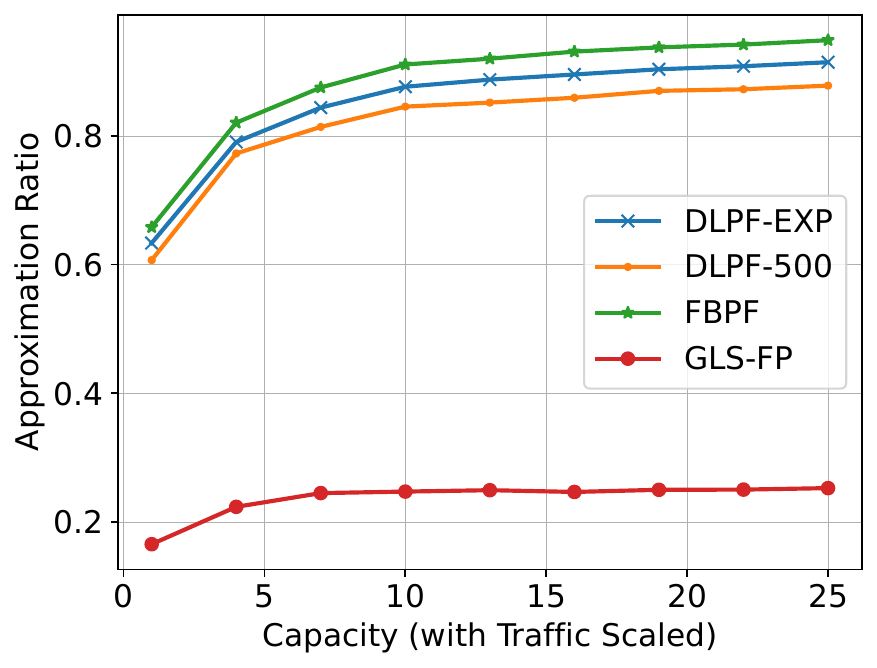}
  \caption{Varying capacity and traffic}
  \label{fig:comparison-b-synthetic}
\end{subfigure}
\begin{subfigure}{.33\textwidth}
  \centering
  \includegraphics[width=0.9\linewidth]{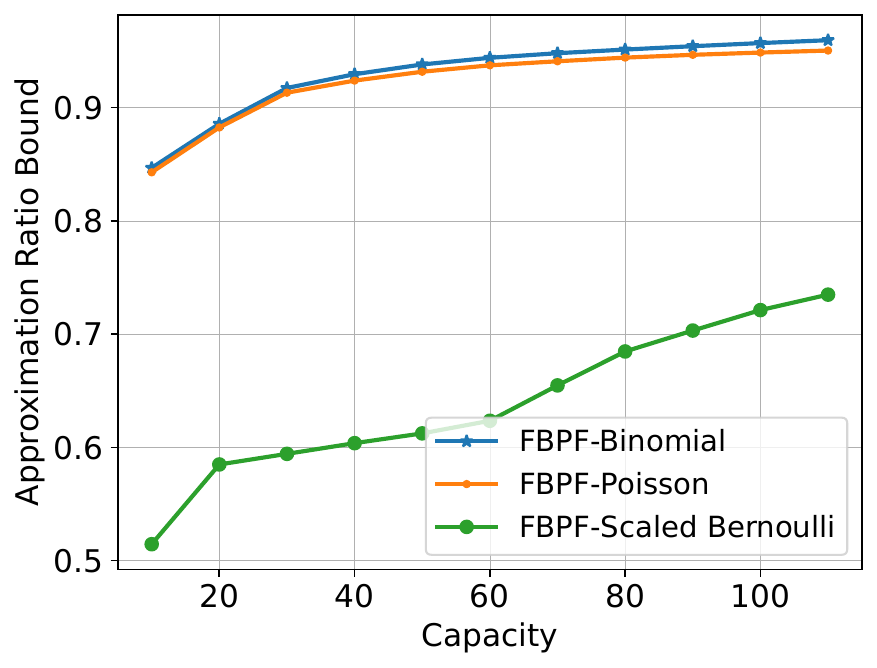}
  \caption{Three packet arrival distributions}
  \label{fig:comparison-diff-distrs}
\end{subfigure}
\caption{Comparison of algorithms in terms of obtained approximation ratios.}
\label{fig:rewards-comparison-synthetic}
\vspace{- 0.1 in}
\end{figure}

In this section, we focus on understanding the performance of our algorithms in real networks with real packet arrival traces. First, we consider 
IBM's network $\G_{IBM}$ \cite{ibmnetwork}, shown in \Cref{fig:comparison-a}. 
To simulate the packet arrivals, we use a trace of IP traffic collected at the University of Cauca \cite{rojas2019consumption}.
We focus on the arrivals during the $10$ minutes of the dataset. As nodes in networks such as $\G_{IBM}$ typically serve multiple clients (e.g. routers of an ISP, serving the traffic of users), we assign the traffic of roughly $1000$ IPs (all the IPs in the trace) to the $18$ nodes of the network. Based on this assignment, we obtain a traffic flow for each ordered pair of nodes in $\G_{IBM}$. For example in $18$ nodes, there are $18 \times 17$ source-destination pairs. Many of these pairs, following the random assignment of IPs result in little to no traffic. As a result, we focus on the $30$ source-destination pairs with the highest traffic. We illustrate three of the final traces of packet arrivals for different source-destination pairs with different characteristics in \Cref{fig:arrival-traces}.
We observe that all three traces have clear non-stationarity and are characterized by sudden bursts of packets. 

\begin{figure}[t]
\centering
\begin{subfigure}{.30\textwidth}
  \centering
  \includegraphics[width=0.95\linewidth]{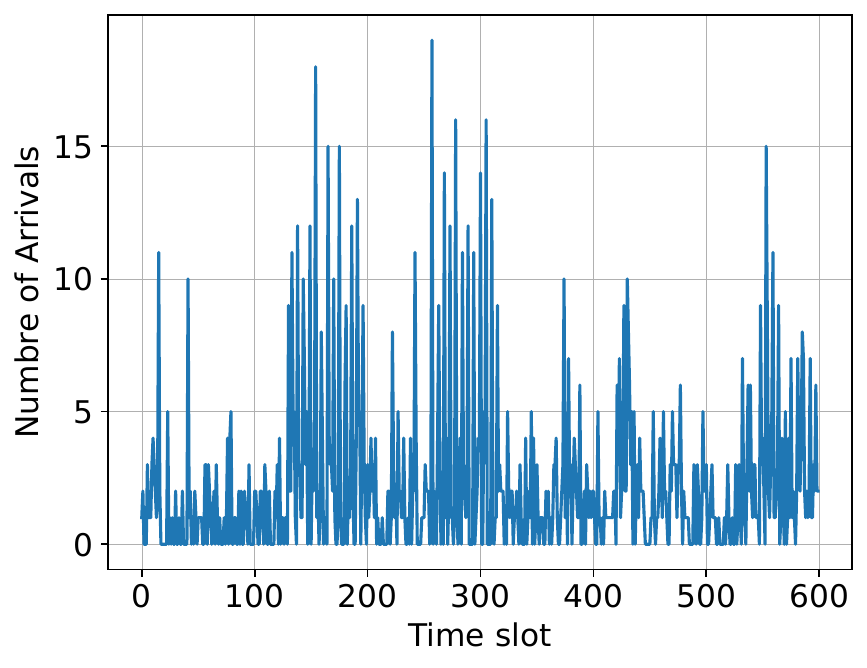}
  \caption{Trace $1$}
  \label{fig:comparison-a}
\end{subfigure}%
\begin{subfigure}{.30\textwidth}
  \centering
  \includegraphics[width=0.95\linewidth]{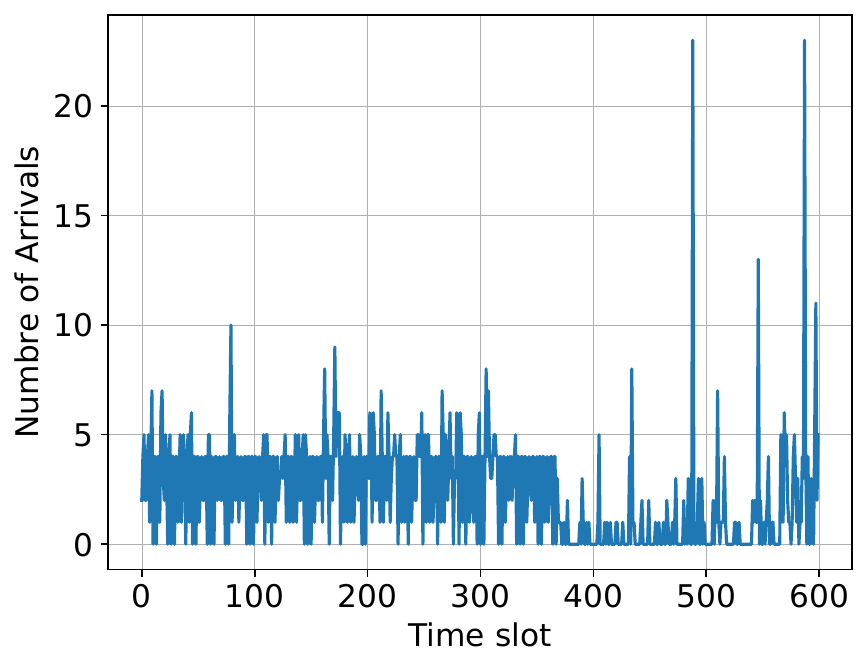}
  \caption{Trace $2$}
  \label{fig:comparison-b}
\end{subfigure}
\begin{subfigure}{.30\textwidth}
  \centering
  \includegraphics[width=0.95\linewidth]{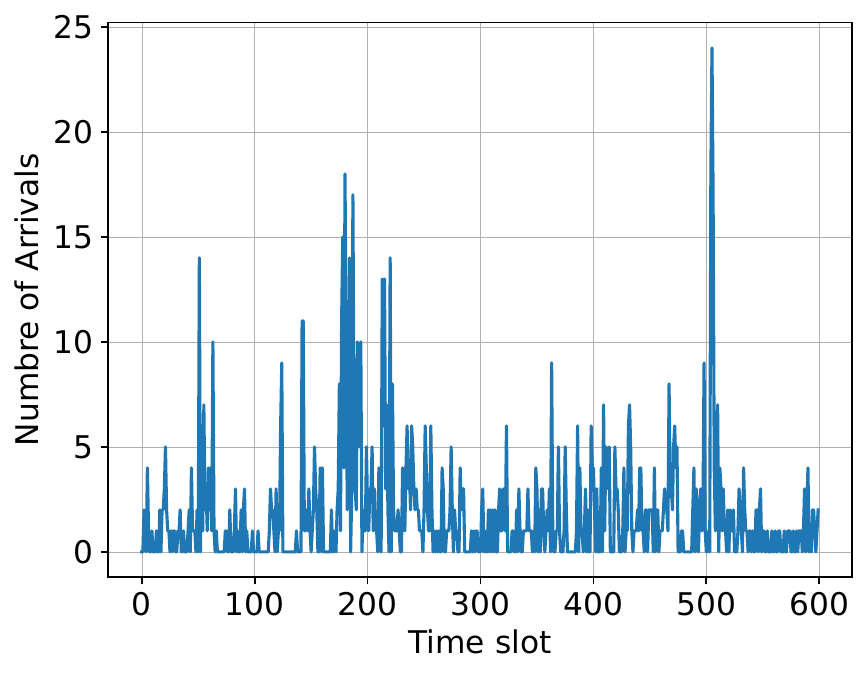}
  \caption{Trace $3$}
  \label{fig:comparison-c}
\end{subfigure}%
\caption{Three different traces of packet arrivals for three source-destination pairs in $\G_{IBM}$.}
\label{fig:arrival-traces}
\vspace{- 0.1 in}
\end{figure}

In our simulations, we assign a random reward uniformly chosen between $(0,1)$ for each source-destination pair and a deadline $10$ time slots to each packet. To further evaluate the impact of traffic intensity and the number of flows, we simulate two cases. In one case, we maintain all top $30$ source-destination pairs, whereas in the second experiment, we maintain a random subset of $15$ source-destination pairs.

As the traffic is non-stationary, we focus our study on DLPF-30, DLPF-100, two variants of our dynamic algorithm DLPF, that update their estimates every $30$, and $100$ time slots respectively (our preliminary experiments included phase lengths larger than $100$, which were found to attain a lower reward compared to phase lengths of $30$ and $100$). We compare the algorithms in terms of their average reward per time slot. The final results for the two experiments are shown in \Cref{fig:comparison-a-realdata} and \Cref{fig:comparison-b-realdata}. As we can see, the algorithms maintain a significant gain over GLS-FP, and the gain grows as the traffic intensity increases and the problem becomes more challenging.

Finally, we verified the good performance of our methods over an additional network. We used the Hibernia Atlantic (Canada) network from \cite{ibmnetwork}, which is a network of smaller size compared to $\G_{IBM}$. Under similar simulation configurations but over the new topology (and the mapping of the IPs on the new nodes), we simulated the three methods, showing the results in \Cref{hibernia-atlantic-network}. Both of our methods are outperforming GLS-FP in this topology as well.

 \begin{figure}[t]
 \centering
 \begin{subfigure}{.3\textwidth}
   \centering
   \includegraphics[width=0.90\linewidth]{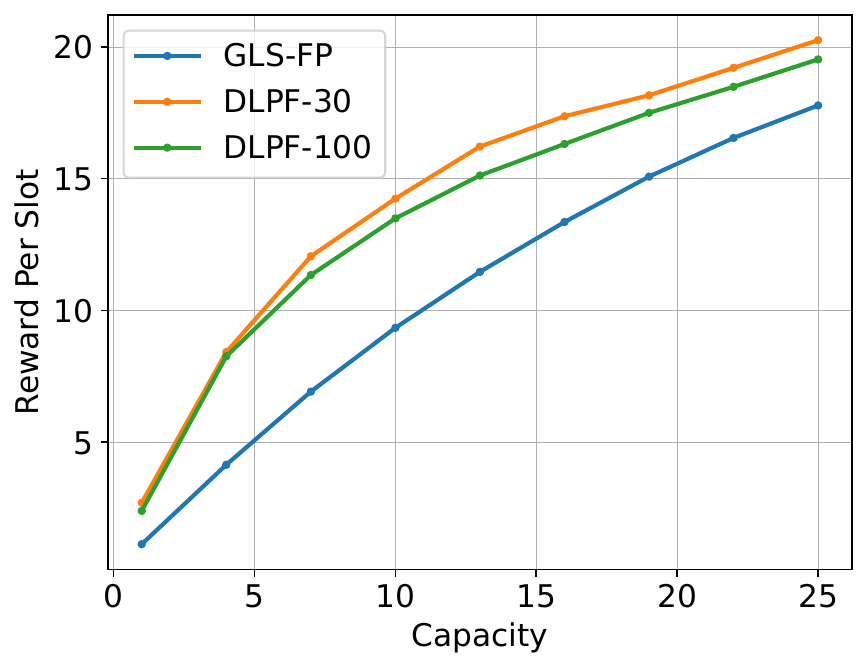}
   \caption{$\G_{IBM}$ with 15 source-destination pairs}
   \label{fig:comparison-a-realdata}
 \end{subfigure}%
 \hfill
 \begin{subfigure}{.3\textwidth}
   \centering
   \includegraphics[width=0.90\linewidth]{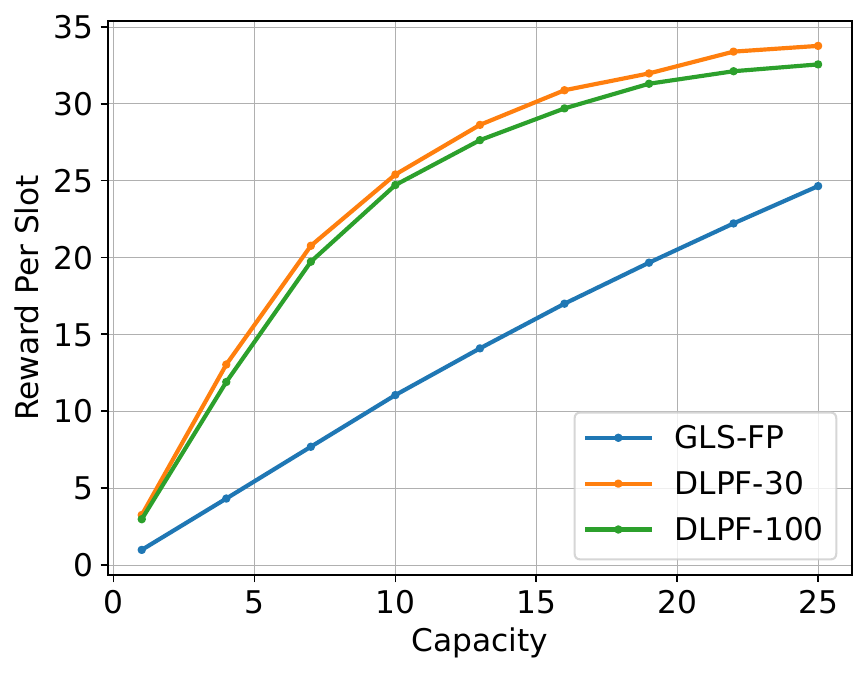}
   \caption{$\G_{IBM}$ with 30 source-destination pairs}
   \label{fig:comparison-b-realdata}
 \end{subfigure}
 \hfill
  \begin{subfigure}{.3\textwidth}
   \centering
   \includegraphics[width=0.90\linewidth]{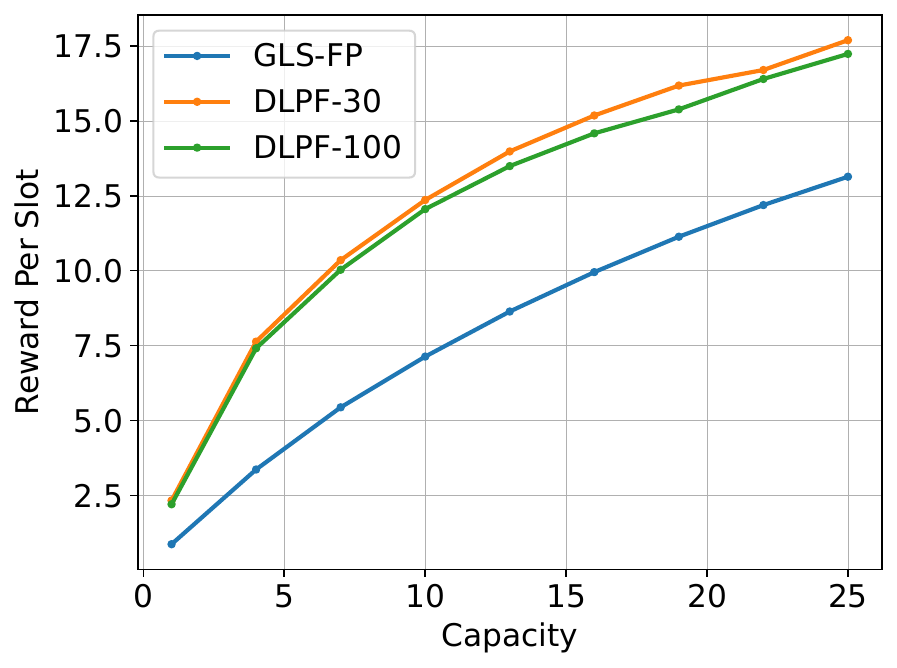}
   \caption{Hibernia Atlantic Network with 15 source-destination pairs}
   \label{hibernia-atlantic-network}
 \end{subfigure}
 \caption{Performance comparison (in terms of average reward per time slot) using real traffic on $\G_{IBM}$ (\Cref{fig:comparison-a-realdata}, \Cref{fig:comparison-b-realdata}) and the Hibernia Atlantic Network (\Cref{hibernia-atlantic-network}).}
 \label{fig:real-data-results}
 \vspace{- 0.1 in}
 \end{figure}

\section{Conclusions}
This paper shows that it is possible to design algorithms that provide near-optimal $(1-\epsilon)$ approximation ratio for the problem of scheduling deadline-constrained multihop traffic in the stochastic setting. Our result requires minimum link capacity $\CMN= \Omega(\log(L/\epsilon)/\epsilon^2)$ to guarantee such performance. This is a significant improvement over the prior results in the worst-case traffic setting. Our techniques are general and can be applied under different distributional assumptions. An interesting future work is to investigate more efficient and distributed ways of solving the LP \FS, e.g., by using distributed primal-dual methods.

\bibliographystyle{ACM-Reference-Format}
\bibliography{citations}

\appendix
\section{Remaining Proofs}
\subsection{\texorpdfstring{Proofs for
\Cref{randomized-scheduling-offline}}{}} \label{rem-proofs-offline}
\begin{proof}[Proof of \Cref{expected-random-integer-lemma}]
Consider possibly non-stationary, general packet arrival distribution, and without loss of generality assume $\abart>0$.
Further, consider the following solution for the expected instance \EI (\dref{EIT:first}-\dref{EIT:last}): 
$\vEI t = \E{\sumajt \vRIO }/{\abart}$
, where, $\vRIO$ is the solution to the random instance \RI, and the expectation is taken with respect to the randomness in the definition of \RI. This solution  is feasible for \EI. Indeed, first, note that constraints \dref{EIT:one-choice} are satisfied:
\[\sumk \vEI t = \sumk \frac{ \E{\sumajt \vRIO}} {\abart} = 
\frac{ \E{ \sumajt \sumk \vRIO}} {\abart} \overset{(\alpha)}{\leq} \frac{ \E{a_j^t}}{\abart} = 1,
\]
where in $(\alpha)$, we used inequality $\sumk \vRIO \leq 1$ from \Cref{RIT:one-choice}, as $\vRIO$ must be feasible for each random program \RI.
Further, for the capacity constraint \dref{EIT:capacity}, we have 
\begin{align}
&\sumtlast \sum_{j=1}^{\PKT} \abart \sumschedl x_{j k}^{\trl} = 
\sumtlast \sum_{j=1}^{\PKT} \sumschedl  \abart x_{j k}^{\trl}  \\
&=\sumtlast \sum_{j=1}^{\PKT} \sumschedl \abart \frac{ \E{\sumajtrl \vRIOtrl}} {\abart} \\
& 
=\E{\sumtlast \sum_{j=1}^{\PKT} \sumschedl \sumajtrl \vRIOtrl}\overset{(a)}{\leq} \E{\C_\ell}=\C_\ell,
\end{align}
where in $(a)$ we used  \dref{RIT:capacity} due to the feasibility of $\vRIO $. Finally, the non-negativity constraints \dref{EIT:last} are trivially satisfied.
The objective value of the proposed solution $\vEI t$ is then:
\[
\sumtotarrv{t} w_j \abart \sumk \vEI t = \E{\sumtotarrv{t} w_j \sumajt \sumk \vRIO } = \ERI. 
\]
Therefore, we proposed a feasible solution to the expected instance which has objective value $\ERI$.  This implies that for the optimal objective value $W_\EI$ we have: $W_{\EI}\geq \ERI$.
\end{proof}

\subsection{\texorpdfstring{Proofs for
\Cref{randomized-scheduling-online}}{}} \label{proofs-online}

\begin{proof}[Proof of \Cref{number-of-samples}]\label{lemma5}
For a packet type $j$, let $A_j$ denote the event of estimating within sufficient accuracy its probability: $\{ \abar (1-\epsilon_0)\leq \ahat \leq  \abar (1+\epsilon_0) \}$. We will show that $P(\bigcap_{j=1}^\PKT A_j) \geq 1-\epsilon$, or equivalently $P(\bigcup_{j=1}^{\PKT} A_j^c) \leq \epsilon$. Using the Hoeffding bound \cite{hoeffding1963probability} we have for the error probability for a fixed $j$: 
\[\Pr(A_j^c)=\Pr(\abar (1-\epsilon_0)\leq \ahat \leq  \abar(1+\epsilon_0))^c \leq 2 \exp\left(-2 \abar ^2 \epsilon^2_0 T/a_{\max}^2 \right) \overset{(a)}{\leq} \frac {\epsilon}{\PKT},\]
where $(a)$ holds for \[\TC\geq \frac{a_{\max}^2}{\amin^2 \epsilon^2_0 2} \log \left ( \frac {2 \PKT}{\epsilon}\right ).\]
The result is obtained by union bound over all $j$.
\end{proof}

\begin{proof}[Proof of \Cref{static-approximation}]
The proof is conditioned on the assumption of  accurate estimates $\{ \abar (1-\epsilon_0)\leq \ahat \leq  \abar (1+\epsilon_0) \}$ which holds with probability at least $(1-\epsilon)$ based on \Cref{number-of-samples}. In this case, first, consider the solution $\mathbf x^\star$ to $\EST$ (identified from $\mathbf f^\star$, with $\epsilon=0$). Let us refer to the modification of optimization $\EST$ with probabilities $\ahat$ as $\hat E_S$. Then $\frac{\mathbf x^\star}{1+\epsilon_0}$ is feasible for $\hat E_S$. Indeed for the capacity constraints we have: 
\[
\sumTT \sumpkt \ahat w_j \sumKK \frac{x_{jk}^\star}{1+\epsilon_0} 
\overset{(a)}{\leq} 
\sumTT \sumpkt \abar w_j  \sumKK x_{jk}^\star
\overset{(b)}\leq \C_\ell
\]
where $(a)$ follows from \Cref{number-of-samples}, and $(b)$ follows from the feasibility of 
$\mathbf x^\star$ for $\EST$. The remaining constraints follow trivially as they do not depend on $\abar$.
Since $\frac{\mathbf x^\star}{1+\epsilon_0}$ is a feasible solution to
$\hat E_S$, we have, due to the optimality of $\mathbf {\hat x}$ for $\hat E_S$:

\begin{equation}
\sumpkt w_j \ahat \sumKK \hat x_{jk}
\geq 
\sumpkt w_j \ahat \sumKK \frac{x_{jk}^\star} {1+\epsilon_0} \label{lemma:first-important-inequality}
\end{equation}

We proceed by also showing that $\mathbf{\tilde x} = (1-\epsilon_0) \mathbf{\hat x}$ is feasible for $\EST$. Indeed:
\[
\sumTT \sumpkt \abar w_j \sumKK (1-\epsilon_0) \hat x_{jk}
\leq 
\sumTT \sumpkt \ahat w_j  \sumKK \hat x_{jk}
\leq C_\ell.
\]
In the following, we obtain a relation between the objective value for $\tilde {\mathbf x}$ and $\hat {\mathbf x}$, and then between $\hat {\mathbf x}$ and $\mathbf x^\star$, to conclude the Lemma.
\begin{equation}
\frac{1+\epsilon_0}{1-\epsilon_0} 
\sumpkt w_j \abar \sumKK \tilde x_{jk}
= (1+\epsilon_0) 
\sumpkt w_j \abar \sumKK \hat x_{jk} 
\geq 
\sumpkt w_j \ahat \sumKK\hat x_{jk} \label{lemma:second-inequality}
\end{equation}
Further, we have 
\begin{equation}
\sumpkt w_j \ahat \sumKK \frac{x_{jk}^\star} {1+\epsilon_0}
\geq 
\frac{1-\epsilon_0}{1+\epsilon_0}
\sumpkt w_j \abar \sumKK x_{jk}^\star \label{lemma:third-important-inequality}
\end{equation}
Using Inequality \dref{lemma:first-important-inequality} and Inequality \dref{lemma:third-important-inequality} and plugging into inequality \dref{lemma:second-inequality}  
we get:
\[
\sumpkt w_j \abar \sumKK \tilde x_{jk} \geq \frac{(1-\epsilon_0)^2}{(1+\epsilon_0)^2}
\sumpkt w_j \abar \sumKK x_{jk}^\star
\]
Further note that 
$\frac{(1-\epsilon_0)^2}{(1+\epsilon_0)^2} \geq 1-4\epsilon_0$, yielding the final result. 
\end{proof}

\begin{proof}[Proof of \Cref{online-theorem}]
Recall that the time horizon of length $\TC$ is divided into $\PHASES+1$ phases, with $\epsilon 2^{\PHASES}=1$. For simplicity we assume that $\epsilon$ is such that $\PHASES = \log (1/\epsilon)$ is an integer.
For each phase we ignore any transmissions of the first $\DMAX$ slots (i.e., as if they were idle) for the purpose of the analysis. The actual performance of the algorithm can only be better than that.
Then, any phase $\phi$ (with $\phi\geq 1$) is split into the idle period $I_\phi$ of duration $\DMAX$ and an active period of duration $T_i$. The total active time is $\TCP=\sum_{\phi=1}^{\PHASES} \TC_\phi$ and the total time is $\TC=\TC^\prime + \DMAX \PHASES$.
Recall that $\TC_{0}=\TC_1=\epsilon \TC^\prime$ and for the remaining values $\TC_\phi = \epsilon \TCP 2^\phi$. 

Let $\mu = \frac{a_{\max}^2}{2 \amin^2}$. We saw in \Cref{number-of-samples}, that using $\mu \frac{\log(2 \PKT/\epsilon)}{\epsilon_0^2}$ samples, we obtain a $1-4\epsilon_0$ approximation to $\EST$ w.p. at least $(1-\epsilon)$. 
Since we are scheduling with the corresponding scaled flow probabilities, similarly to the proof of \Cref{main-theorem-offline} in \Cref{technique-static}, we have a loss of factor $3\epsilon$. Combining the various approximation ratio losses, we obtain the following fraction of the optimal value for each period of accuracy $\epsilon_0$:
$(1-4\epsilon_0)(1-3\epsilon)(1-\epsilon)\geq (1-4\epsilon_0) (1-4 \epsilon)$ 

More specifically, in our case, in phase $\phi$ of total duration $\TC_\phi$ ($\phi\geq1$), the algorithm can use all prior phases $0,1,\cdots, \phi-1$ of total duration $\sum_{i=0}^{\phi-1} \TC_i = \TC_0 + \TC_0 (2^{\phi}-1)= \TC_0 2^{\phi}=\TC_\phi$.
Therefore, during phase $\phi$, we use: $\epsilon_0 := \epsilon_{\phi} = \sqrt{ \mu \frac{\log(2 \PKT/\epsilon)}{\TC_\phi}}$, and we obtain a fraction $(1-4\epsilon)(1-4\epsilon_\phi)$ of the optimal.
Overall, we can identify the approximation ratio of \Cref{randomized-scheduling-online} by analyzing the quantity.
\begin{equation}
(1-4\epsilon)\sum_{\phi=1}^{\PHASES} \frac{\TC_\phi}{\TC}(1-4\epsilon_{\phi}), \label{approximation-ratio-to-analyze}
\end{equation}
Indeed this can be seen as we are obtaining  $(1-\epsilon_\phi)(1-4\epsilon)$ approximation in phase $\phi$, for a fraction $\TC_\phi/\TC$ of the total time. We therefore have:
\begin{align*}
& \sum_{\phi=1}^{\PHASES} \frac{\TC_\phi}{\TC}(1-4\epsilon_{\phi}) = 
\sum_{\phi=1}^{\PHASES} \frac{\TC_\phi}{\TC}-
\sum_{\phi=1}^{\PHASES} \frac{\TC_\phi}{\TC} 4\epsilon_{\phi} \geq 
\sum_{\phi=1}^{\PHASES} \frac{\TC_\phi}{\TC}-
\sum_{\phi=1}^{\PHASES} \frac{\TC_\phi}{\TC^\prime}4\epsilon_{\phi} \\
&\overset{(a)}\geq 
(1-\epsilon) \sum_{\phi=1}^{\PHASES} \frac{\TC_\phi}{\TC^\prime}-
\sum_{\phi=1}^{\PHASES} \frac{\TC_\phi}{\TC^\prime} 4\epsilon_{\phi} 
\overset{(b)}=(1-\epsilon)^2-
\sum_{\phi=1}^{\PHASES} \frac{\TC_\phi}{\TC^\prime} 4\epsilon_{\phi}  \geq 
(1-2\epsilon)-
\sum_{\phi=1}^{\PHASES} \frac{\TC_\phi}{\TC^\prime} 4\epsilon_{\phi}. 
\end{align*}

To see $(a)$, note that by assumption: 
$\TC \geq \DMAX (1 + \frac 1 \epsilon) \PHASES$. We then have:
\[
\TCP = \TC-\DMAX \PHASES \geq  \TC- \frac{\TC}{1+1/\epsilon} = \frac{\TC}{1+\epsilon} \geq \TC(1-\epsilon).
\]
from which $(a)$ follows by the reciprocal inequality after dividing with $(1-\epsilon)$.
Further, $(b)$ holds due to $\sum_{\phi=1}^\PHASES {\TC}_\phi/{\TCP} = \frac{\TCP-\TC_0}{\TCP}=\frac{\TCP-\epsilon \TCP}{\TCP} = 1-\epsilon$.
We will proceed by showing  $\sum_{\phi=1}^{\PHASES} \frac{\TC_\phi} {\TC} \epsilon_{\phi} \leq \epsilon$ which will give us a bound:
\begin{equation}
\sum_{\phi=1}^{\PHASES} \frac{\TC_\phi}{\TC}(1-4\epsilon_{\phi})  \geq 1-6\epsilon. \label{related-bound}
\end{equation}


Indeed, we have
\[
\sum_{\phi=1}^{\PHASES} \TC_\phi \epsilon/\TCP 
\overset{(a)}=
\sqrt{\mu \log(2\PKT/\epsilon) \epsilon/\TCP} \sum_{\phi=1}^{\PHASES} \sqrt{2}^{(\phi-1)} \overset{(b)}\leq 
\frac{\sqrt{\mu \log(2\PKT/\epsilon) \epsilon/\TCP} }{\sqrt \epsilon (\sqrt 2 -1 )} = 
 \frac{\sqrt{ \mu \log(2 \PKT/\epsilon)/\TC^\prime}}{\sqrt 2 -1}\overset{(c)}\leq \epsilon.  \]
$(a)$ is from definition of $\epsilon_\phi$ and $\TC_\phi$. In $(b)$, we used the fact that $\sum_{\phi=1}^{\PHASES} \sqrt{2}^{\phi-1}=\frac{\sqrt {2}^{\PHASES}-1}{\sqrt 2 - 1} =  (1/\sqrt{\epsilon}-1)/(\sqrt 2 -1)\leq \frac {1}{\sqrt \epsilon (\sqrt 2 -1)}$
since
$2^{\PHASES/2} \sqrt{\epsilon} = 1$.
Finally, $(c)$ holds under the assumption that $\TC> \frac{1}{(\sqrt 2 -1)^2} \mu \log(2 \PKT/\epsilon)/\epsilon^2 + \DMAX \log (1/\epsilon)$
which implies $\TC^\prime>  
\frac{1}{(\sqrt 2 -1)^2} \mu \log(2 \PKT/\epsilon)/\epsilon^2$.
Therefore, from \dref{related-bound} and due \dref{approximation-ratio-to-analyze}, we obtain a fraction
$(1-6\epsilon)(1-4\epsilon)\geq (1-10\epsilon)$ of $\RI$.
\end{proof}
\section{General Distributions with stationary arrival rates}\label{dependent-packet-arrivals}
In this section we consider general distributions on the arrival processes $\{\ajt\}$ with stationary arrival rates, i.e., $\avgpacks_{j}^t \equiv \avgpacks_j$. First, we provide a preliminary definition which allows us to extend \Cref{main-theorem-offline}.
\begin{definition} \label{arbitrary-distr-remark}
Consider a nonnegative integer-valued random variable $Z$ with distribution $\mathcal D$, and suppose $Z$ is bounded by a constant $M$. Consider a decomposition of $Z$ as the sum of $N$ independent random variables $\{Z_i: Z_{i} \leq M_i\}$:
\[
Z = \sum_{i=1}^{N} Z_{i}.
\]
We define the dependency-degree of the decomposition as  $\max _i M_i$. Then, define the dependency-degree of distribution $\mathcal D$, as the minimum dependency-degree of all decompositions for $\mathcal D$.
\end{definition}

In particular, the Binomial and Bernoulli distributions, both have dependency-degree  $1$ (i.e., $M_i=1$). A scaled Bernoulli distribution which is either $0$ or $M$, has the maximum possible dependency-degree of $M$.
Further, it is easy to see that the dependency-degree of the sum of two independent random variables is equal to the maximum dependency degree of the two random variables. 

We now state \Cref{main-theorem-offline-dependencies}, which generalizes \Cref{main-theorem-offline}.
\begin{theorem}
Consider stationary arrival processes $\{a_{j}^t\}$, with the property that each packet type's total traffic within a fixed time window, $\sumtlast a_{j}^t$, has at most dependency degree $D$ (for all $j$ and $t$). Then,
    given $\epsilon \in (0,1/3)$, $\ALGOFF$ provides $(1-3\epsilon)$-approximation to $\RI$ when $\CMN \geq 2 D \left(\frac{ 1+\epsilon}{\epsilon}\right)^2 \log (L/\epsilon)$ and $\TC\geq \frac{2 \DMAX^2} {\epsilon}$.  \label{main-theorem-offline-dependencies}
\end{theorem}
The minimum link-capacity threshold for a given approximation ratio, now depends on the dependency degree of the distribution of the total arrivals for the packet type in a window of $(\DMAX+1)$ consecutive time slots. In other words, a lower dependency-degree results in better performance guarantees, given $\CMN$.
\begin{remark}
Due to \Cref{main-theorem-offline-dependencies}, we can obtain guarantees for distributions with dependencies across time slots.
For example, if $a_{j}^t\in \{0,1\}$, and we allow the arrivals of a packet type  at a given time, to depend on the arrivals of that type at different time slots, then there can be a maximum of $\DMAX+1$ packets depending on each other, which results in dependency degree of $D=(\DMAX+1)$. An example of such a distribution is given in \cite{lee2021generalized}.
\end{remark}


For the proof of \Cref{main-theorem-offline-dependencies}, the analysis remains mostly identical to the case of Bernoulli or Binomial arrivals. However we need to generalize \Cref{small-drop-chance-lemma} as follows.

\begin{lemma}\label{small-drop-chance-lemma-generalization}
In the case that
$\sumtlast a_{j}^t$ has at most dependency degree $D$, for all $j$ and $t$,
using the forwarding probabilities  $\{f_{j\ell}^{\tau \star}\}$ for scheduling packets (Lines \ref{alg-packet-for}-\ref{alg-packet-for-end} of \Cref{randomized-scheduling-offline}), the probability of a packet being dropped is at most $\epsilon$, if $\CMN \geq 2 D \left(\frac{ 1+\epsilon}{\epsilon}\right)^2 \log \frac L \epsilon$.
\end{lemma}
\begin{proof}
The proof in \Cref{small-drop-chance-lemma} leverages the independence of packet types. More specifically the proof for \Cref{small-drop-chance-lemma} relied on \Cref{concentration-lemma} to show that any link's capacity is unlikely to be exceeded. This was then used to obtain through a union bound the final result. We follow similar reasoning here, however when using \Cref{concentration-lemma} we need to define a set of independent variables, which now cannot be the arrival or not of a packet on the link. 

Suppose all the packets that can arrive on a link $\ell$ are partitioned into $N$ groups. We define the total consumption due to group $i$ as $Z_{g_i}$. We then have $Z_{g_i} \leq D$ by assumption (although the assumption is imposed on the arrivals for each packet type, this is propagated to the arrivals at a link, due to the properties of the dependency-degree). We proceed similarly to \Cref{small-drop-chance-lemma} by analyzing $\ES{Z_{g_i}^2}$ and $\ES{Z_{g_i}}$. However, now, in contrast to \Cref{small-drop-chance-lemma}, it is not true that $\ES{Z_{g_i}^2}=\ES{Z_{g_i}}$. Indeed let us consider a group $g$ with a maximum number of $I$ packets that can arrive in the considered link, and let these arrivals be indicated through variables $Y_1,Y_2,\cdots,Y_I$.
Then 
\[\ES{[Z_{g}^2]}=
\ES{\big(\sum_{i=1}^I Y_i\big)^2}= \ES{\big(\sum_{i=1}^I Y_i + \sum_{i \neq j} Y_i Y_j\big)}  =
\ES{\big[\sum_{i=1}^I Y_i (1+\sum_{j:j\neq i}^I Y_j)\big]} \leq 
D \ES{\big[\sum_{i=1}^I Y_i}]= D \ES {[Z_g]}
\]
Using the above, we prove (based on the notation of \Cref{concentration-lemma}, similarly to \Cref{small-drop-chance-lemma}): 
\[S=\sum_{i=1}^{N} \ES{[Z_{g_i}^2]} \leq D \ES{Z} \leq D \C_{\ell}/(1+\epsilon)=D S_u,\]
where $S_u$ and $\lambda$ are defined as in the proof of \Cref{small-drop-chance-lemma}. Then the exponent in \Cref{concentration-lemma}, for $B=D$, becomes
\[
-\frac{\lambda^2}{2S+2 D\lambda/3} \leq 
-\frac{1} D
\frac{\epsilon^2 S_u^2 }{2S_u+2 \epsilon S_u /3 } \leq 
- \frac{1} D\frac {\epsilon^2 \C_{\ell} }{2(1+\epsilon)^2}
\]
from which we can see that $\C_{\ell}$ needs to scaled this time by a factor of $D$ in order to replicate the results of the proof in \Cref{small-drop-chance-lemma}.
\end{proof}

\section{Non-stationary distributions}
\subsection{Extension to non-stationary arrival rates} \label{non-stationary-extension}
In this section we present details for the non-stationary extension of \Cref{randomized-scheduling-offline}. First, we need to extend \FS, to obtain \FNS. In essence, \FNS is obtained directly from \EI, by omitting Step 2 in \Cref{outline-theorem-offline}. As a result, randomizing according to the probabilities from the solution of \FNS results in improved approximation ratios over using the probabilities obtained from \FS, as we do not need to leverage \Cref{static-is-good-lemma} anymore (which adds a multiplicative factor of $(1-2 \DMAX^2/T)$ in the approximation ratio and thus worsens the quality of the guarantee for small horizons). Alternatively, \FS could be interpreted as an approximation of \FNS. 

The analysis follows similarly to the stationary case, by omitting the intermediate step of \Cref{static-is-good-lemma}. An analogous statement as in \Cref{small-drop-chance-lemma} can then be derived for \FNS. The \LP \FNS is given below:

\begin{subequations}\label{FS-TIME}
\begin{align}
\max_{\mathbf{f}} \quad & 
\sum_{t=1}^\TC \sum_{j=1}^{\PKT} w_{j} \avgpacks_j^t \sum_{\ell \in \AO{s_{j}}} \fl{j} \ell {0t}
\quad(:=\mathrm F_{NS}) \label{FS-obj-TIME} \\ 
\textrm{s.t.} \quad & 
\sum_{\ell \in \AO{s_{j}}} \fl j \ell {0t} \leq 1,\quad \forall j\in \PD, \forall t\in \TC, \label{FS:one-choice-TIME}\\
& 
\sum_{\ell \in \AI{v}} \fl{j}{\ell}{(\trl-1)t}=\sum_{\ell \in \AO{v}} \fl j \ell {\trl t},\quad \forall v\in \V, \trl \in [\DMAX], \forall j\in \PD, \label{FS:time-conserv-TIME}\\
&f_{j\ell}^0=0,\ \ \forall \ell\not \in \AO{s_j}, 
\quad f_{j\ell}^{d_j}=0,\ \ \forall \ell\not \in \AI{d_j}, \label{FS:edge-cases-TIME}
\\
& 
\sum_{\trl=0}^{\DMAX}  \sumpkt \avgpacks_{j}^{t-\trl} f_{j\ell}^{\trl(t-\trl)}\leq \frac{\C_{\ell}}{1+\epsilon}, \quad \forall \ell\in \LNKS, \forall t\in [\TB],
\label{FS:capacity-TIME}\\
& \fl j {\ell} {\trl t} \geq 0,\quad \forall j \in \PD,\forall \ell \in \ELNKS,\forall \trl\in [\DMAX], t\in [\TC]. \quad \label{FS-postv-TIME}
\end{align}
\end{subequations}
In the optimization above, we conventionally consider $f_{j\ell}^{\tau t}=0$ for $t\leq0$ and $t>\TC$ (in order to simplify indexing in \Cref{FS:capacity-TIME}). We now state a theorem for the general non-stationary case.

\begin{theorem}\label{most-general-thm}
Consider general arrival processes $\{a_{j}^t\}$, with the property that each packet type's total traffic within a fixed time window, $\sumtlast a_{j}^t$, has at most dependency degree $D$ (for all $j$ and $t$). Then,
    given $\epsilon \in (0,1/2)$, and scheduling using the forwarding variables from \FNS, yields a $(1-2\epsilon)$-approximation to $\RI$ when $\CMN \geq 2 D \left(\frac{ 1+\epsilon}{\epsilon}\right)^2 \log (L/\epsilon)$.  
\end{theorem}
Note that \Cref{most-general-thm} does not impose any constraints on the time horizon. The use of \FNS for the forwarding variables comes with a cost in terms of computational complexity. In particular, the number of variables and constraints required to solve \FNS now scale with the horizon $\TC$ (which was the reason we considered \FS in stationary environments). This can be mitigated as we discuss next in \Cref{frame-base-extension-appendix}.

\subsection{Frame-Based solutions for non-stationary arrival rates} \label{frame-base-extension-appendix} 
As discussed in \Cref{non-stationary-extension}, in the presence of non-stationary traffic, the number of variables and constraints for solving \FNS scales with horizon $\TC$. However, a frame-based solution can be adopted to mitigate the computational cost of solving \FNS, while offering additional benefits which we discuss below.

In a frame-based construction, we aim to partition horizon $[\TC]$ into consecutive time slots, named ``frames``, of some maximum duration $H$, with $H\ll \TC$. Then, we solve \FNS for each frame independently, and in order to forward packets that have arrived within a frame, we use the forwarding variables obtained from the solution of \FNS at the beginning of that frame. 

Frame-based constructions offer three motivating benefits: (a) We need not know (or estimate) ahead of time the arrival rates for the entire horizon, but only for the next frame. This is particularly relevant to prediction-based methods, since the arrival rates are often predictable in the short-term future, but difficult to predict over longer horizons.
(b) The complexity of solving the \LP is reduced. For example, if we are using an \LP solver with complexity $O(n^{k})$ for some $k>1$, where $n$ the number of variables, then, as $n \propto \TC$, we benefit by solving the problem in frames of size $H$, which results in $\TC/H$ optimization problems, with $n\propto H$ each, and therefore a reduction of the complexity from $O(\TC^k)$ to $O(\frac \TC H  H^k)=O(T H ^{k-1})$. (c) The computational complexity of finding the forwarding variables, is more evenly split across the time horizon, rather than requiring the calculation of the forwarding variables for the entire horizon at the beginning of the execution of the algorithm.

A frame-based method however might lead to a degradation in the approximation ratio obtained, unless it is meticulously designed, or certain assumptions are imposed on the traffic. This is because the scheduling decisions within a frame do not take into consideration the possible arrivals in other frames, despite the fact that these decisions may impact these future frames. It then becomes important to characterize the losses due to a frame construction. Fortunately, these losses can be localized around the boundaries of the frame (the end or the beginning of the frame), since, due to the bound $\DMAX$ on the deadlines, the impact of decisions can be shown to be limited to at most $\DMAX$ slots away from the boundaries of each frame. However, without additional assumptions, due to the arbitrary non-stationary distributions on the arrivals, a fixed-frame construction could have significant losses, if most of the high-weight traffic is concentrated on these boundaries. This can be resolved by either placing certain mild assumptions on the non-stationarity of the traffic, or by constructing a more involved frame construction, with frames of potentially different frame lengths. We illustrate a variable-sized frame construction.

\Cref{frm-constr-lemma} characterizes the losses of using a variable-sized frame construction.
\begin{lemma}\label{frm-constr-lemma}
\FNS admits a frame-based solution, with frames of variable length, of at most $H=(2+4/\epsilon)\DMAX$ time slots each, with
\[
W_{\FNS}^{\mathrm{FRAME}} \geq (1-\epsilon) W_{\FNS} 
\]
\end{lemma}
\begin{proof}
We assume for simplicity that $\TC$ is a multiple of $\DMAX N$ for some $N$. Then, let us partition $\TC$ into groups of $\DMAX$ slots, referred to as mini-frames.  
Further, define cycles $\mathcal C_{1}, \mathcal C_{2}, \cdots$, with each cycle consisting of $N$ consecutive mini-frames. Therefore each cycle has duration $\DMAX N$. 
Assume we are processing some arbitrary cycle, say $\mathcal C_{1}$ without loss of generality.
Solve \FNS on each mini-frame for $\mathcal C_{1}$ and $\mathcal C_2$. Suppose the optimal solutions within each mini-frame are of reward $W[1],W[2],\cdots, W[N]$ for the first cycle, and $W^{\prime}[1],\cdots, W^{\prime}[N]$ for the second cycle. 
Find $i$ with minimum reward for the first cycle $W_{min}= W[i]$, and similarly $i^{\prime}$ for the second cycle. These minimum reward frames can be assigned to be the boundaries of the frame and intuitively they should result in limited losses. Indeed we define the frame to be the set of consecutive mini-frames between (but not including) mini-frames $i$ and $i^{\prime}$. We argue that the loss in each cycle due to omitting these mini-frames will be limited.
Consider mini-frame $i$, which is part of cycle $\mathcal C_{i}$. 
Since our frame-based solution will not schedule packets in this mini-frame, it will suffer a loss from arrivals in that mini-frame. However outside the mini-frame, the remaining reward will be at least as large as the non-frame-based solution. As a consequence of the definition of $i$, we can further argue that these losses are limited. In particular, the loss cannot be higher than a fraction $2/(N-1)$ of the reward in the cycle. Indeed, if that was the case, $i$ wouldn't have been the minimum reward mini-frame in the cycle.
Ignoring all transmissions in this mini-frame can result in a loss of reward of the optimal policy due to transmissions in this mini-frame, which cannot be higher than a fraction $2/(N-1)$ of the reward in the entire frame. If $2/(N-1)=\epsilon$, that is, for $N=1+2/\epsilon$, we have an $\epsilon$ loss in each cycle.
The above scheme, only requires knowing or estimating the distribution of arrivals over the incoming two cycles. \Cref{changing-frame-construction} illustrates the construction through an example.
\end{proof}

\begin{figure}[t]
\centering
\includegraphics[width=0.6\linewidth]{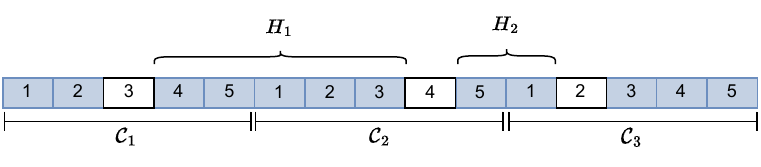}
\caption{A frame-based construction, where frames are defined as a sequence of mini-frames between two mini-frames of minimum reward. The figure indicates an example with $N=5$ mini-frames in each cycle. In the first cycle, mini-frame $3$ has minimum reward $W[i]$. In the second cycle, mini-frame $4$ has minimum reward $W^\prime[i^\prime]$.
}
\label{changing-frame-construction}
\end{figure}

\subsection{Reducing Complexity for Periodic Arrival Rates} \label{periodic-traffic-extension}

\begin{figure}[t]
\centering
\begin{subfigure}{.40\textwidth}
  \centering
  \includegraphics[width=0.95\linewidth]{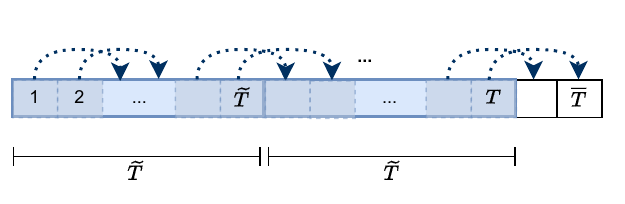}
  \caption{Range of capacity constraints for \EI}
  \label{fig:before-symmetrize-per}
\end{subfigure}%
\begin{subfigure}{.40\textwidth}
  \centering
  \includegraphics[width=0.95\linewidth]{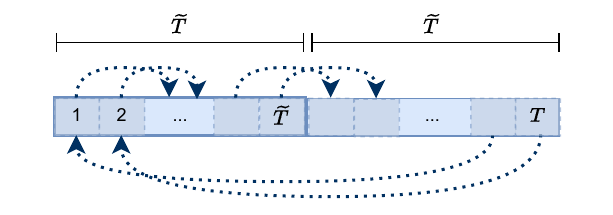}
  \caption{Modified range for \EIB}
  \label{fig:after-symmetrize-per}
\end{subfigure}
\caption{We consider the same example as in \Cref{fig:illustrate-symmetrizing}, with $\DMAX=2$, but now, rather than a stationary traffic, we assume periodic traffic with $\TC=2\widetilde {\TC}$, illustrated in \Cref{fig:before-symmetrize-per}. The broad symmetry of the problem breaks again at $1,2, \TC-1, \TC$. After the modification shown in \Cref{fig:after-symmetrize-per}, the first $\widetilde {\TC}$ slots are indistinguishable from the final $\widetilde {\TC}$ slots in \EIB. For example the arrivals rates impacting slot $1$ are identical to the arrival rates impacting slot $1+\widetilde {\TC}$. This implies a symmetry in shifts by $\widetilde {\TC}$ slots, which yields equivalent solutions following by shifting a solution in the time index by $\widetilde {\TC}$ (similarly to the proof of \Cref{static-is-good-lemma}).}
\label{fig:illustrate-symmetrizing-periodic}
\vspace{- 0.1 in}
\end{figure}

Here, we provide details on how a periodic set of forwarding variables can be used to obtain a near-optimal solution for \RI. For simplicity, it is assumed that $\widetilde \TC$ divides the horizon $\TC$. Then, to argue that a periodic solution to \EI is near optimal, we need to provide an argument similar to that in the proof of \Cref{static-is-good-lemma} in a stationary setting (i.e., generalize for $\widetilde {\TC}\geq 1$). 

Indeed, it is not hard to verify that the modified program, \EIB is now symmetric with respect to shifts in the index of the variables by $\widetilde {\TC}$. Refer to \Cref{fig:illustrate-symmetrizing-periodic} for an illustration of the point. Then, any general optimal solution to $\EI$, can be made periodic, by noticing that a shift of all time-indices in the solution of \EI by $\widetilde {\TC}$ results in another optimal solution. Then, averaging all shifts, yields a solution that only depends on the timeslot within the period, and hence, is periodic. From that argument as in \Cref{static-is-good-lemma}, we infer that a periodic near-optimal solution exists for \EI, and hence a simplification similar to \EST can be obtained, which converted in a flow form, yields the following program $\mathrm F_{NSP}$.

\begin{subequations}\label{FS-TIME-P}
\begin{align}
\max_{\mathbf{f}} \quad & 
\sum_{t=1}^{\widetilde{\TC}} \sum_{j=1}^{\PKT} w_{j} \avgpacks_j^t \sum_{\ell \in \AO{s_{j}}} \fl{j} \ell {0t}
\quad(:=\mathrm F_{\mathrm{NSP}}) \label{FS-obj-TIME-P} \\ 
\textrm{s.t.} \quad & 
\sum_{\ell \in \AO{s_{j}}} \fl j \ell {0t} \leq 1,\quad \forall j\in \PD, \forall t\in [\widetilde\TC], \label{FS:one-choice-TIME-P}\\
& 
\sum_{\ell \in \AI{v}} \fl{j}{\ell}{(\trl-1)t}=\sum_{\ell \in \AO{v}} \fl j \ell {\trl t},\quad \forall v\in \V, \trl \in [\DMAX], \forall j\in \PD, \label{FS:time-conserv-TIME-P}\\
&f_{j\ell}^0=0,\ \ \forall \ell\not \in \AO{s_j}, 
\quad f_{j\ell}^{d_j}=0,\ \ \forall \ell\not \in \AI{d_j}, \label{FS:edge-cases-TIME-P}
\\
& 
\sum_{\trl=0}^{\DMAX}  \sumpkt \avgpacks_{j}^{(t-\trl)\mathrm{mod} \widetilde \TC} f_{jl}^{\trl((t-\trl)\mathrm{mod} \widetilde \TC) }\leq \frac{\C_{l}}{1+\epsilon}, \quad \forall \ell\in \LNKS, t\in [\widetilde \TC],
\label{FS:capacity-TIME-P}\\
& \fl j {\ell} {\trl t} \geq 0,\quad \forall j \in \PD,\forall \ell \in \ELNKS,\forall \trl\in [\DMAX], t\in [\TC]. \quad \label{FS-postv-TIME-P}
\end{align}
\end{subequations}

\section{Relation between flows and route-schedules}\label{flow-relation}

Recall that in this work we randomize over route-schedules for each packet (e.g. following our analysis in \EST). In order to assign probabilities to all route-schedules compactly, we reformulate \EST, inspired by ideas from the maximum flow and multicommodity flow problem \cite{magnanti1993network}. 
First, we need to define a new graph over which we can refer to flows. To do this, we initially associate route-schedules to regular routes in a time-expansion of the original network, such as the one indicated in \Cref{time-expanded-appendix}. In the time-expansion each link is copied $\DMAX+1$ times (with each copy representing the age at which the packet can be scheduled over the link) and each node is copied $\DMAX+2$ (representing the location of the packet at different ages).

Then, all route schedules $\mathcal K_j$ correspond to paths in the time-expansion from node $(s_j:0)$ to $(z_j:\DMAX+1)$. However, from standard results in the network flow literature \cite{magnanti1993network}, we can combine a composition of weighted paths in any graph, into a flow, and vice versa, we can decompose a weighted flow into a set of paths \cite{goldberg1998beyond}. In particular, by making sure that the weights are appropriately scaled (i.e., the total weight of the weighted paths is at most one), we can interpret the weights as probabilities of scheduling over each path.

Applied to our problem, this allows us to work with either variables directly on route-schedules $x_{jk}$ (\EST), or their flow-inspired corresponding forwarding variables $f_{j\ell}^{\trl}$. For instance, note that the example of forwarding-variables  defined in \Cref{fig:flow-assignment-example} can also be viewed as a flow
over the time-expansion in \Cref{time-expanded-appendix}.

\subsection{From forwarding variables to route-schedules}\label{iterative-process-flow-to-route}
As in the network flow literature \cite{magnanti1993network,raghavan1987randomized}, there are multiple ways to assign forwarding variables to probabilities over route-schedules. In the paper we presented a hop-by-hop approach, which uses directly the intuitive definition of a forwarding variable. An alternate mapping however could be through an iterative process that in each step finds a route-schedule of maximum possible probability, subtracts the probability assigned to it from the flow, and repeating the process. This results in exhausting for each packet type at least one edge of the time-expansion per iteration. Since there is a maximum of $(\DMAX+1) |\ELNKS|$ edges in the time-expansion, this results in a limited number of route-schedules with positive probability per packet type.


\begin{figure}[t]
\centering
\includegraphics[width=0.3\linewidth]{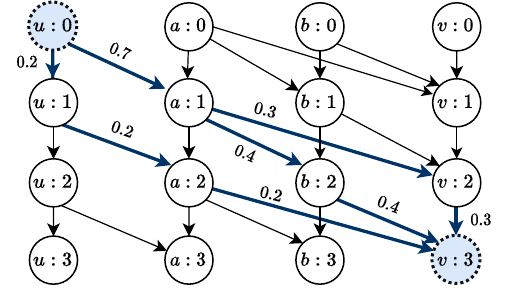}
\caption{A time-expansion of the network in \Cref{fig:flow-assignment-example} with the packet type considered in \Cref{fig:flow-assignment-example}. Note that there are only $3$ route-schedules for that packet type, and these correspond to all paths from $u:0$ to $v:3$ (recall that the deadline was $2$). For example path $(u:0,u:1,a:2,v:3)$ corresponds to route schedule $k=[(u,u)\quad(u,a)\quad(a,v)]$. Each route-schedule can be assigned a probability through the forwarding variables indicated in both \Cref{fig:flow-assignment-example} as well as the time-expansion. Note however that in the time-expansion, the forwarding variables consist also of a valid flow as in the max-flow problem \cite{magnanti1993network}.
}
\label{time-expanded-appendix}
\end{figure}

\section{ADDITIONAL SIMULATIONS}\label{additional-simulations}

\subsection{Synthetic Datasets}
\textbf{Grid Network $4\times 4$.} We simulated the performance of our algorithm in a grid network with $16$ nodes, in a $4\times4$ arrangement. 
In the grid topology, a node which is not at an edge or corner of the grid, is connected with $4$ nodes (top,right,bottom,left nodes). Corner nodes only have 2 connections instead. For the simulation parameters, we used $\PKT=30$, and randomly assigned sources and destinations, as earlier, and weights randomly chosen in $\{10,\cdots,100\}$.  The results are shown in \Cref{fig:grid-various} for different $\DMAX$ and $\abar$ ranges (i.e. each $\abar$ is uniformly distributed in that range).

\begin{figure}[t]
\centering
\begin{subfigure}{.30\textwidth}
  \centering
  \includegraphics[width=0.95\linewidth]{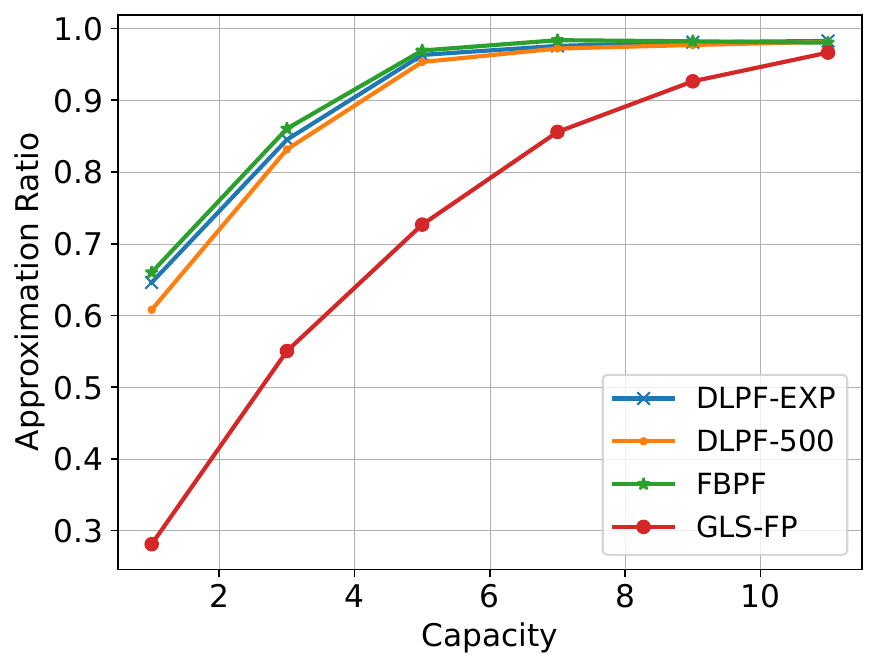}
  \caption{$\DMAX=5,\abar \in (0,1)$}
  \label{grid-network-sim}
\end{subfigure}
\begin{subfigure}{.30\textwidth}
  \centering
  \includegraphics[width=0.95\linewidth]{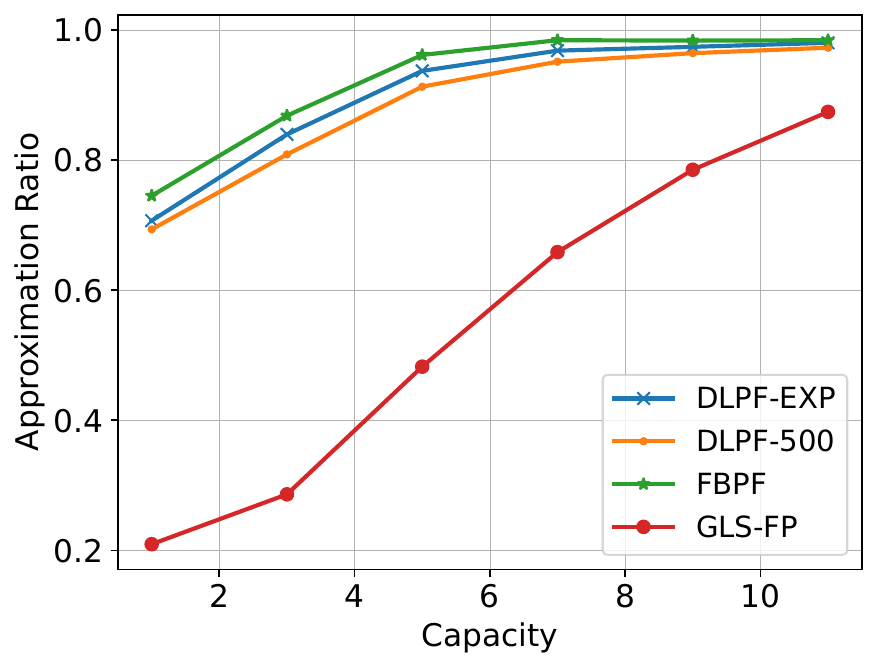}
  \caption{$\DMAX=5,\abar \in (0.8,1)$}
  \label{grid-network-sim-dead7}
\end{subfigure}%
\begin{subfigure}{.30\textwidth}
  \centering
  \includegraphics[width=0.95\linewidth]{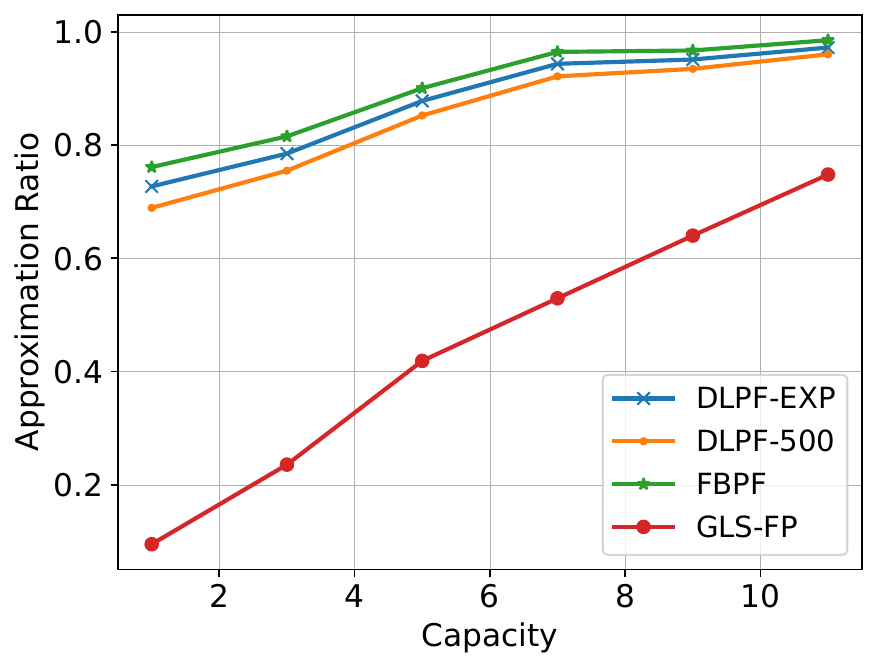}
  \caption{$\DMAX=7,\abar \in (0.8,1)$}
  \label{grid-network-sim-3}
\end{subfigure}
\caption{Performance evaluations for Grid Network.}
\label{fig:grid-various}
\vspace{- 0.1 in}
\end{figure}

\noindent\textbf{Circle Network with $8$ nodes.} We simulate two network topologies and show the results in \Cref{fig:circle-results}. First, we use a simple circle topology with links $(1,2), (2,3), \cdots (7,8), (8,1)$ (with the results shown in \Cref{bare-circle}). We subsequently introduce a short connection connecting $(1,5)$ (with the results shown in \Cref{circle-with-short}). We used $\DMAX=15$ and the arrival probabilities were chosen uniformly in $\abar\in (0.5,1.0)$. Weights are generated similarly to the Grid Network and we used $\PKT=20$.

\begin{figure}[t]
\centering
\begin{subfigure}{.30\textwidth}
  \centering
  \includegraphics[width=0.95\linewidth]{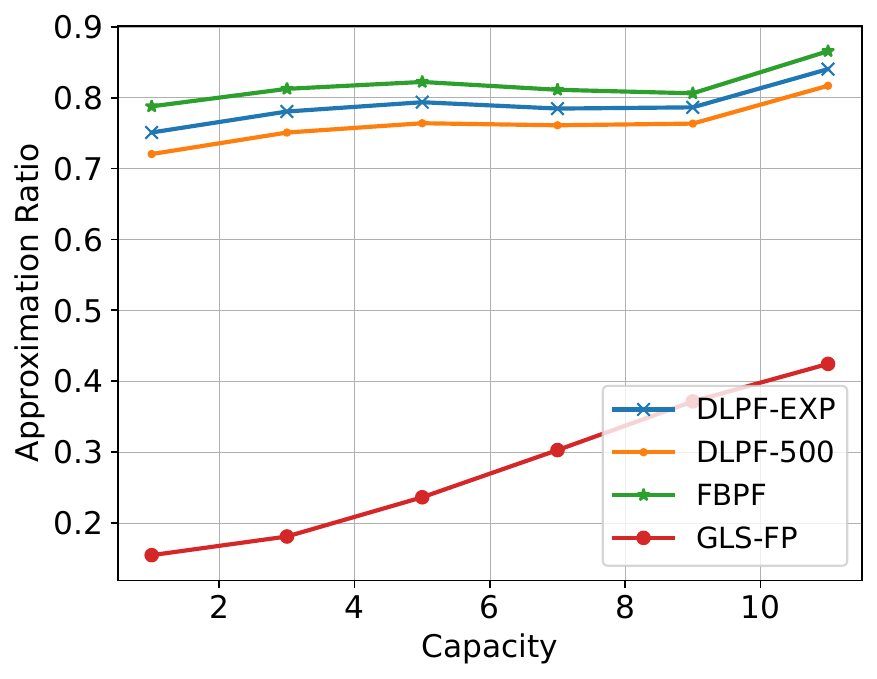}
  \caption{Simple circle topology}
  \label{bare-circle}
\end{subfigure}%
\begin{subfigure}{.30\textwidth}
  \centering
  \includegraphics[width=0.95\linewidth]{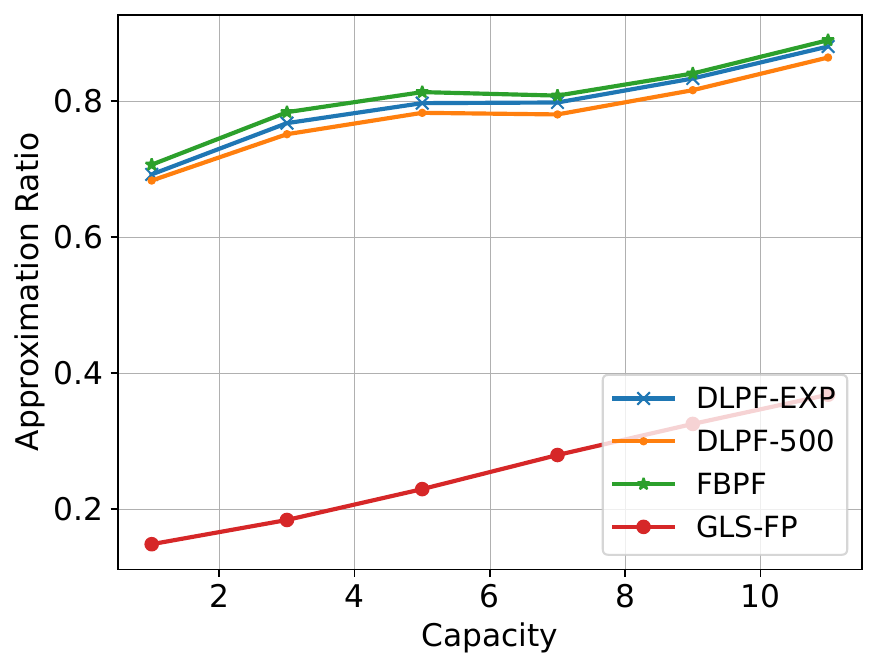}
  \caption{Circle with shortcut}
  \label{circle-with-short}
\end{subfigure}
\caption{Performance evaluations for Circle Topologies.}
\label{fig:circle-results}
\vspace{- 0.1 in}
\end{figure}

\end{document}